\documentclass{article} \usepackage[T1]{fontenc}
\DeclareMathAlphabet{\monbb}{U}{bbold}{m}{n}
\usepackage{amsmath,amsfonts,amssymb, amsthm,url}
\usepackage{mathrsfs}
\usepackage{stmaryrd} 
\usepackage{graphics} \usepackage{color}
\usepackage{tikz}
\usepackage{enumitem}
\usepackage{multicol}
\usepackage{color,booktabs}

\usepackage{marvosym,fancybox}

\renewcommand{\epsilon}{\varepsilon}

\newcommand{\N}{\mathbb{N}} 
\newcommand{\Z}{\mathbb{Z}} \newcommand{\Q}{\mathbb{Q}}

\newcommand{\W}{\mathbb{W}} \newcommand{\A}{\mathcal{A}}
\newcommand{\B}{\mathcal{B}}

\newcommand{\Lmu}{L_{\mu}}
\newcommand{\mul}{\Lambda_{\mu}}

\newcommand\subsubsubsection[1]{\paragraph{#1}}

\newcommand{\st}{\ |\ }

\theoremstyle{plain} \newtheorem{theorem}{Theorem}[section]
\newtheorem{prop}[theorem]{Proposition}
\newtheorem{corollary}[theorem]{Corollary}
\newtheorem{lemma}[theorem]{Lemma} 
\newtheorem{claim}[theorem]{Claim}

\newenvironment{claimproof}[1]{\par\noindent\textit{Proof of the claim:}\space#1}{\leavevmode\unskip\penalty9999 \hbox{}\nobreak\hfill\quad\hbox{$\blacksquare$}}

\theoremstyle{definition} 
\newtheorem{definition}{Definition}[section]

\theoremstyle{remark} 
\newtheorem{remark}{Remark}[section]
\newtheorem{ex}{Example}[section] 

\newcounter{factcount}[theorem]

\newcommand{\bprf}[1][Proof:]{\begin{list}{}    {\setlength{\leftmargin}{0.5em} \setlength{\rightmargin}{0em}  \setlength{\listparindent}{1em}}   \item {\em \hspace{-1em}  #1  }}
\newcommand{\eprf}{\end{list}}

\title{$\mu$-Limit Sets of Cellular Automata from a Computational Complexity Perspective}

\author{{Laurent Boyer},{Martin Delacourt},{Victor Poupet},{Mathieu Sablik},\\{Guillaume Theyssier}}
\date{}

\begin{document}
\maketitle
\footnote{Research partially supported by the FONDECYT Postdoctorado Proyecto 3130496 and by grant 'Agence Nationale de la Recherche ANR-09-BLAN-0164'}

  \begin{abstract}
    This paper concerns $\mu$-limit sets of cellular automata: sets of
    configurations made of words whose probability to appear does not
    vanish with time, starting from an initial $\mu$-random
    configuration. More precisely, we investigate the computational
    complexity of these sets and of related decision problems. Main
    results: first, $\mu$-limit sets can have a $\Sigma_3^0$-hard
    language, second, they can contain only $\alpha$-complex
    configurations, third, any non-trivial property concerning them is
    at least $\Pi_3^0$-hard. We prove complexity upper bounds, study
    restrictions of these questions to particular classes of CA, and
    different types of (non-)convergence of the measure of a word
    during the evolution.
  \end{abstract}

\section{Introduction}

A cellular automaton (CA) is a complex system defined by a local rule which acts synchronously and uniformly on the configuration space. These simple models have a wide variety of different dynamical behaviors, in particular interesting asymptotic behaviors. 

In the dynamical systems context, it is natural to study the limit set of a cellular automaton: it is defined as the set of configurations that can appear arbitrarily far in time. This set captures the longterm behavior of the CA and has been widely studied since the end of the 1980s. Given a cellular automaton, it is difficult to determine its limit set. Indeed it is undecidable to know if it contains only one configuration~\cite{Kari-1992} and more generally, any nontrivial property of limit sets is undecidable~\cite{Kari-1994}. Another problem is to characterize which subshift can be obtained as limit set of a cellular automaton. This was first studied in detail by Lyman Hurd~\cite{Hurd-1987}, and significant progress have been made since~\cite{Maass-1995,Formenti-Kurka-2007} but there is still no characterization. The notion of limit set can be refined if we consider the notion of attractor~\cite{Hurley-1990-1,Kur}. 

However, these topological notions do not correspond to the empirical point of view where the initial configuration is chosen randomly, that is to say chosen according a measure $\mu$. That's why the notion of $\mu$-attractor is introduced by~\cite{Hurley-1990-2}. Like it is discussed in~\cite{km-2000} with a lot of examples, this notion is not satisfactory empirically and the authors introduce the notion of $\mu$-limit set. A $\mu$-limit set is a subshift whose forbidden patterns are exactly those, whose probabilities tend to zero as time tends to infinity. This set corresponds to the configurations which are observed when a random configuration is iterated.

As for limit sets, it is difficult to determine the $\mu$-limit set of a given cellular automaton, indeed it is already undecidable to know if it contains only one configuration~\cite{bpt-2006}, and as for limit sets, every nontrivial property of $\mu$-limit sets is undecidable~\cite{Delacourt-2011}. In \cite{Boyer-Delacourt-Sablik-2010}, it was shown that large classes of subshifts such as transitive sofic subshifts can be realized as $\mu$-limit sets. 

This paper aims at pushing techniques already used in \cite{Boyer-Delacourt-Sablik-2010,Delacourt-2011} to their limits in order to characterize the complexity of $\mu$-limit sets and associated decision problems. The main contribution is to show that the complexity of $\mu$-limit sets can be much higher than that of limit sets. This fact may seem counter-intuitive given that limit sets take into acount worst-case initial conditions whereas $\mu$-limit sets restrict to $\mu$-typical initial configurations, thus excluding possibly complex behaviors. However our proofs show that: first, some self-organization can be achieved from random initial configurations in order to initiate more or less arbitrarily chosen computations; second, the probabilistic conditions involved in the definition of $\mu$-limit sets allow in fact to encode more complexity in the decision problem of whether a word is accepted in the $\mu$-limit language or not.

 This article, after a section dedicated to definitions, is organized
 as follows:
\begin{itemize}
\item in Section~\ref{sec:constr} we give the detail of a generic
  construction we will use many times. It is similar to the ones
  in~\cite{Boyer-Delacourt-Sablik-2010,Delacourt-2011} but presented
  here as a ready-to-use tool (see Theorem~\ref{thm:maintheorem}).
\item  in Section~\ref{sec:complex} we give bounds on the complexity of the language of the $\mu$-limit set, which in general case is $\Sigma_3$-hard, then we show that this bound can be reached. We also give a cellular automaton whose $\mu$-limit set contains only $\alpha$-complex configurations.
\item in Section~\ref{sec:properties}, we deal with properties of $\mu$-limit sets. First we show that every nontrivial property is at least $\Pi_3$-hard. Then we investigate the complexity of $\mu$-nilpotency for different classes of CA.
\item in Section~\ref{sec:convergence} we discuss convergence issues. In particular the type of convergence: general limsup, Cesaro mean limit, simple convergence. We also show evidence of some late (non-recursive) convergence phenomena.
\end{itemize}

In the recent work~\cite{Hellouin-Sablik-2013}, similar constructions (with fairly different implementation details) are used, mainly to prove reachability results concerning limit probability measures obtained by iterating a CA from simple initial measures. Among other results, the set of measures that can be obtained as a simple limit is completely characterized, and moreover, it is proven that any set of measures following a necessary computability  condition and a natural ``topological'' condition can be achieved as a set of limit points of a sequence of measures obtained by iteration of a CA from a simple initial measure. This gives an interesting complementary point of view to the one adopted in the present paper, the link being that the $\mu$-limit set is the closure of the union of supports of limits points of the sequence of measures obtained by iterations. However, the translation of these results into the setting of $\mu$-limit sets is somewhat artificial, and, in any case, it does not give the complexity lower bounds established in this paper.

\section{Definitions}

\subsection{Words and Density}
\label{sub:words}
For a finite set $Q$ called an \emph{alphabet}, denote $Q^{\ast}= \bigcup_{n\in\N} Q^n$ the set of all finite words over $Q$. The \emph{length} of $u = u_0u_1\dots u_{n-1}$ is $|u| = n$. 
 We denote $Q^{\Z}$ the set of \emph{configurations} over $Q$, which are mappings from $\Z$ to $Q$, and  for $c\in Q^{\Z}$, we denote $c_z$ the image of $z\in \Z$ by $c$. Denote $\sigma$ the shift map, i.e. the translation over the space of configurations: $\forall c\in Q^{\Z}, \forall z \in \Z, \sigma(c)_z=c_{z+1}$. For $u\in Q^{\ast}$ and $0\leq i\leq j<|u|$, define the \emph{subword} $u_{[i,j]} =u_iu_{i+1}\dots u_j$; this definition can be extended to a configuration $c\in Q^{\Z}$ as $c_{[i,j]} =c_ic_{i+1}\dots c_j$ for $i,j\in\Z$ with $i\leq j$. The \emph{language} of a configuration $c\in Q^{\Z}$ is defined by $$L(c)=\{ u \in Q^{\ast}: \exists i \in\Z \textrm{ such that } u = c_{[ i , i + | u | - 1 ] }\}.$$ This notion extends naturally to any set of configuration ${S\subseteq Q^\Z}$ by taking the union. An important category of sets of configurations is that of \textit{subshift}. A subshift is a set of configuration which is translation invariant and closed for the product topology on $Q^\Z$. Equivalently, they are sets defined by languages; a set $S\subseteq Q^\Z$  is a subshift if there is a language $L$ of \textit{forbidden words} defining $S$, \textit{i.e.}
\[S = \{c : L(c)\cap L=\emptyset\}.\]

Subshifts are the core objects of symbolic dynamics~\cite{Lind-Marcus-1995}. Among the different kinds of subshifts, we will consider \textit{effective subshifts}, \textit{i.e.} those such that the forbidden language can be chosen recursively enumerable.\\

For every $u \in Q^{\ast}$ and $i\in\Z$, define the \emph{cylinder} $[u]_i$ as the set of configurations containing the word $u$ in position $i$ that is to say $[u]_i = \{c\in Q^{\Z} : c_{[i,i+|u|-1]} = u\}$. If the cylinder is at the position $0$, we just denote it by $[u]$.

For all $u,v\in Q^{\ast}$ define $|v|_u$  the \emph{number of occurences} of $u$ in $v$ as:
$$|v|_u=\textrm{card}\{i\in[0,|v|-|u|] : v_{[i,i+|u|-1]}=u\}$$
(in particular ${|v|_u=0}$ as soon as ${|u|>|v|}$).

For finite words $u,v\in Q^{\ast}$, if $|u|<|v|$, the density of $u$ in $v$ is defined as $d_v(u)=\frac{|v|_u}{|v|-|u|}$. For a configuration $c\in Q^{\Z}$, the \emph{density} $d_c(v)$ of a finite word $v$ is:
\begin{displaymath}
d_c(v)=\limsup_{n\to +\infty} \frac{|c_{[-n,n]}|_v}{2n+1-|v|}.
\end{displaymath}

These definitions can be generalized for a set of words $W\subset
Q^{\ast}$, we write $|u|_W$ and $d_c(W)$. We can give similar
definitions for semi-configurations (indexed by $\N$) too.

We will also use the classical notion of density of a subset
${X\subseteq\Z}$ of integers and denote it simply $d$:
\[d(X) = \limsup_{n\to +\infty} \frac{|X\cap\{-n,\ldots,n\}|}{2n+1}\]

\begin{definition}[Growing computable sequence]
  
  A sequence $w=(w_i)_{i\in\N}$ of finite words on the alphabet $Q$ is
  a \emph{growing computable sequence} when:
  \begin{itemize}
  \item $\lim_{i\to\infty}|w_i|=\infty$;
  \item there exists a Turing machine that computes $w_i$ when given the input $i$.
  \end{itemize}
  Denote $\W(Q)$ the set of growing computable sequences on
  alphabet $Q$. For any $w\in\W(Q)$, we define the associated language
  of persistent words : 
  \[L_w=\{u\in Q^*, d_{w_i}(u)\not\to_{i\to\infty} 0\}.\]
\end{definition}

The following lemma shows that we can produce the persistent language
of any given growing computable sequence by another growing
computable sequence where we have a precise control on time and space
resource needed for the computation of each word of the sequence.

\begin{lemma}
  \label{lem:incrseq}
  Let $T$ and $S$ be computable functions from $\N$ to itself which have the following
  properties:
  \begin{itemize}
  \item $T(i)>>i$ and $i>>S(i)>>\log(i)$;
  \item the time complexity of both $T$ and $S$ are $o(T)$;
  \item the space complexity of $T$ is at most $S$ and that of $S$ is $o(S)$.
  \end{itemize}
  Consider any growing computable sequence $w=(w_i)_{i\in\N}$. Then
  there exists another growing computable sequence
  $w'=(w'_i)_{i\in\N}$ and a Turing machine $\phi$ (with possibly several heads
  and tapes) such that:
  \begin{itemize}
  \item $L_w=L_{w'}$;
  \item $\phi$ computes $w'_i$ on input $i$ in time at most $T(i)$ and
    space at most $S(i)$ (for large enough $i$).
  \end{itemize}
\end{lemma}
\begin{proof}
  Let $\phi_0$ be a Turing machine producing $w_j$ on input $j$.  We
  can suppose without loss of generality that the time ${T_0(j)}$
  spent by $\phi_0$ to produce $w_j$ on input $j$ verifies:
  \[T_0(j+1)\geq 2\cdot T_0(j)\]
  This can be obtained by artificially slowing $\phi_0$ if necessary
  (on input $j$, recompute 2 times step $j-1$ before doing the real work to
  produce $w_j$). 
  We now sketch the behavior of $\phi$ on input $i$:
  \begin{itemize}
  \item $\phi$ has an output tape initialized with the empty word;
  \item it also initializes a space marker at position $S(i)$ and
    precomputes $\lambda(i-1)$ and $\lambda(i)$ ($\lambda$ is a function to be
    precised later, but smaller than $T$ and as easy to compute);
  \item it simulates $\phi_0$ on each successive entry $j\leq i$
  \item at each step of $\phi_0$ it increments some step counter and
    check that it is less than $\lambda(i)$  and that everything still fits within space
    $S(i)$:
    \begin{itemize}
    \item if not it stops with the current output written on the
      output tape;
    \item if it is OK, it goes on;
    \end{itemize}
  \item when $\phi_0$ reaches an halting state on input $j$, it copies the output produced
    by $\phi_0$ on the output tape and then check that the step
    counter is less than $\lambda(i-1)$:
    \begin{itemize}
    \item if it is the case it cleans the working tape of $\phi_0$,
      and start a new simulation of $\phi_0$ on input $j+1$;
    \item if it is not the case it stops and such outputs $w_j=\phi_0(j)$.
    \end{itemize}
  \end{itemize}
  Both the counter incrementation routine and the halting state
  routine above take time $o(i)$ because it is just a matter of doing
  a constant number of erasing/copying/comparing/incrementing words of
  length at most $S(i)$. So if we take ${\lambda(i)}$ smaller than
  ${(T(i)-i)/i}$ we are guaranteed that $\phi$ halts in time at most
  $T(i)$ using space at most $S(i)$ (for $i$ large enough). We also
  want $\lambda(i)$ to grow slowly, precisely such that:
  \[\lambda(i+1)<2\cdot\lambda(i-1).\]

  Then, by construction, $\phi$ always outputs some $w_i$ or the empty
  word. First, if $\phi(i)$ produces $w_j$ then by construction
  $\phi(i+1)$ produces either $w_j$ or $w_{j+1}$. Indeed if $\phi$
  produces $w_j$ on input $i$ it is because:
  \begin{itemize}
  \item either ${\lambda(i-1)<T_0(j)}$ and therefore
    ${\lambda(i+1)<2\cdot T_0(j)\leq T_0(j+1)}$ so that $\phi(i+1)$
    also produces $w_j$;
  \item or ${\lambda(i)<T_0(j+1)}$ and therefore $\phi(i+1)$ produces
    either $w_j$ (in the case where ${\lambda(i+1)<T_0(j+1)}$) or $w_{j+1}$
    (if ${\lambda(i+1)\geq T_0(j+1)}$).
  \end{itemize}
 because $\phi(i)$
  produced $w_j$ with step counter at most $\lambda(i)$, but could not
  produce $w_{j+1}$, so $\phi(i+1)$ either has time to produce
  $w_{j+1}$ and outputs that (by the halting conditions) or is short
  of time and keeps the previous successfull output which is $w_j$. 

Second, for any $j$,
  there must be some large enough $i$ such that $\phi(i)$ produces
  some $w_{j'}$ with ${j'\geq j}$ (precisely, if $i$ is large enough so that $\phi_0(j)$
  halts in less than $\lambda(i)$ steps).  Therefore, the sequence
  $w'$ produced by $\phi$ is, after some finite prefix of empty words,
  of the form:
  \[\underbrace{w_j,\ldots,w_j}_{\text{finite$>0$}},\underbrace{w_{j+1},\ldots,w_{j+1}}_{\text{finite$>0$}},\underbrace{w_{j+2},\ldots,w_{j+2}}_{\text{finite$>0$}},\ldots\]
  We deduce that $L_w=L_{w'}$.
\end{proof}

\subsection{Cellular Automata}

\begin{definition} [Cellular automaton]

A \emph{cellular automaton (CA)} is a triple $\A=(Q_{\A},r_{\A},\delta_{\A})$ where $Q_{\A}$ is a finite set called \emph{set of states} or \emph{alphabet}, $r_{\A}\in \N$ is the \emph{radius} of the automaton, and $\delta_{\A}:Q_{\A}^{2r_{\A}+1}\to Q_{\A}$ is the \emph{local rule}.\\
 
 The configurations of a cellular automaton are the configurations over $Q_{\A}$.
 A global behavior is induced and we will denote $\A(c)$ the image of a configuration $c$ given by: $\forall z\in \Z, \A(c)_z=\delta_{\A}(c_{z-r},\dots,c_z,\dots,c_{z+r})$. Studying the dynamic of $\A$ is studying the iterations of a configuration by the map $\A:Q_{\A}^{\Z}\to Q_{\A}^{\Z}$.
\end{definition}

When there is no ambiguity, we will write $Q$, $r$ and $\delta$ for $Q_{\A}$, $r_{\A}$, $\delta_{\A}$.

\newcommand\allalphabets{\mathcal{Q}}
In this paper, to avoid artificial set-theoretical technicalities, we
fix some countable set ${\allalphabets=\{q_0,q_1,q_2,\ldots\}}$ and adopt the convention that
all cellular automaton alphabets we consider are subsets of
$\allalphabets$. This allows us to speak about the set of all
cellular automata, or the set of all sets of configurations.

A state $a\in Q_{\A}$ is said to be \emph{permanent} for a CA $\A$ if for any $u,v\in Q_{\A}^r$, $\delta (uav)=a$. It is said to be \emph{quiescent} if $\delta (a^{2r+1})=a$.

\subsection{Measures}
\newcommand\cyl[2]{\ensuremath{[{#1}]_{#2}}}
\newcommand\Mes[1]{\ensuremath{\mathcal{M}({#1})}}

We denote by $\Mes{Q^\Z}$ the set of Borel probability measures on $Q^\Z$. By Carath\'eodory extension theorem, Borel probability measures are characterized by their value on cylinders. A measure is given by a function $\mu$ from cylinders to the real interval $[0,1]$ such that $\mu(Q^\Z) = 1$ and
\[\forall u\in Q^*, \forall z\in\Z, \quad \mu(\cyl{u}{z}) = \sum_{q\in Q}\mu(\cyl{uq}{z}) = \sum_{q\in Q}\mu(\cyl{qu}{z-1})\]

A measure $\mu$ is said to be \emph{translation invariant} or \emph{$\sigma$-invariant} if for any measurable set $E$ we have $\mu(E)=\mu(\sigma(E))$.

In addition, $\mu$ is \emph{$\sigma$-ergodic} if for any $\sigma$-invariant measurable set $E$ we have $\mu(E)=0$ or $\mu(E)=1$.
Finally, we say $\mu$ \emph{has full support} if $\mu([u])>0$ for any word $u$.

A $\sigma$-invariant measure $\mu$ is \emph{computable} if there exists some computable ${f:Q^\ast\times\Q\rightarrow\Q}$ (where $Q$ is the set of states) with
\[\forall \epsilon>0,\forall u\in Q^\ast,\ \bigl|\mu([u])-f(u,\epsilon)\bigr|\leq \epsilon\]

The simplest and most natural class of computable and
$\sigma$-invariant measures is that of
\emph{Bernoulli measures}: they correspond to the case where each cell
of a configuration is chosen independently according to a common fixed
probability law over the alphabet.
\begin{definition}[Bernoulli measure]
For an alphabet $Q$, a \emph{Bernoulli measure} is a measure $\mu$
such that:
\[\forall u\in Q^\ast, \forall i\in\Z, \mu([u]_i)=\prod_{q\in
  Q}\mu([q]_0)^{|u|_q}.\]
The state probabilities $\mu([q]_0)$ are called the
\emph{coefficients} of $\mu$. $\mu$ has full support if all coefficients are non-null.

The \emph{uniform Bernoulli measure} $\mu_0$ is the Bernoulli measure
whose coefficients are all equal, equivalentely it is defined by: $$\forall u\in Q^*, i\in \Z, \mu_0([u]_i)=\frac{1}{|Q|^{|u|}}$$
\end{definition}

For a CA $\A=(Q,r,\delta)$ and $u\in Q^*$, we denote for all $t\in \N$, $\A^t\mu([u])=\mu\left(\A^{-t}([u])\right)$.

\begin{definition}[Generic configuration]
A configuration $c$ is said to be \emph{weakly generic} for an alphabet $Q$  and a measure $\mu$ if there exists a constant $M$ such that, for any  word $u\in Q^*$, $\frac{1}{M}\mu([u])\leq d_c(u)\leq M\mu([u])$.  If, moreover, any word has density $\mu([u])$, the configuration is said to be \emph{generic}.
\end{definition}

\begin{remark}
 The set of weakly generic configurations has measure $1$ in $Q^{\Z}$. Which means that a configuration that is randomly generated according to measure $\mu$ is a generic configuration.
\end{remark}

\subsection{$\mu$-Limit Sets}\label{sec:mulimites}

A $\mu$-limit set is a subshift associated to a cellular automaton and a probability measure \cite{km-2000}. It is defined by their language as follows.

\begin{definition}[Persistent set]
For a CA $\A$, define the \emph{persistent set} $L_{\mu}(\A)\subseteq Q^*$ by: $\forall u \in Q^*$:
 $$u\notin L_{\mu}(\A) \Longleftrightarrow \lim_{t\rightarrow\infty}\A^t\mu([u]_0)=0.$$
 
 Then the \emph{$\mu$-limit set} of $\A$ is $\Lambda_{\mu}(\A)=\left\{c\in Q^{\Z}:L(c)\subseteq L_{\mu}(\A) \right\}$.
\end{definition}

\begin{remark}
Two $\mu$-limit sets are therefore equal if and only if their languages are equal.
\end{remark}

\begin{definition}[$\mu$-nilpotency]
A CA $\A$ is said to be $\mu$-nilpotent if $\Lambda_{\mu}(\A)=\{a^{\Z}\}$ for some $a\in Q_{\A}$ or equivalently $L_{\mu}(\A)=a^*$.
\end{definition}

The question of the $\mu$-nilpotency of a cellular automaton is proved undecidable in \cite{bpt-2006}. The problem is still undecidable with CA of radius $1$ and with a permanent state. 

\begin{definition}[Set of predecessors]
Define the set of predecessors at time $t$ of a finite word $u$ for a CA $\A$ as $P^t_{\A}(u)=\left\{v\in Q^{|u|+2rt}: \A^t([v]_{-rt})\subseteq [u]_0\right\}$.
\end{definition}

The following lemma translates the belonging to the $\mu$-limit set in terms of density in images of a weakly generic configuration. 

\begin{lemma}
\label{lem:generic}
Given a CA $\A$, a $\sigma$-invariant measure $\mu\in\Mes{Q^\Z}$ and a finite word $u$, for any weakly generic configuration $c$:\\
\begin{center}$u\notin L_{\mu}(\A)$ $\Longleftrightarrow$  $\lim\limits_{t\to +\infty}d_{\A^t(c)}(u) = 0$
\end{center}
\end{lemma}
\begin{proof}
Let $M$ be such that, for any  word $u\in Q^*$, $\frac{1}{M}\mu([u])\leq d_c(u)\leq M\mu([u])$.

$$d_{\A^t(c)}(u)  =  d_c(P^t_{\A}(u))= \sum_{v\in P^t_{\A}(u)}d_c(v)$$
\begin{eqnarray*}
\sum_{v\in P^t_{\A}(u)}\frac{1}{M}\mu([v]) & \leq\  d_{\A^t(c)}(u)\  \leq & \sum_{v\in P^t_{\A}(u)}M\mu([v])\\
\frac{1}{M}\sum_{v\in P^t_{\A}(u)}\mu([v]) & \leq\  d_{\A^t(c)}(u)\  \leq & M\sum_{v\in P^t_{\A}(u)}\mu([v])\\
\frac{1}{M}\mu(\A^{-t}([u])) & \leq\  d_{\A^t(c)}(u)\  \leq & M\mu(\A^{-t}([u]))\\
\frac{1}{M} \A^{t}\mu([u])& \leq\  d_{\A^t(c)}(u)\  \leq & M\A^{t}\mu([u])
\end{eqnarray*}
This concludes the proof.

\end{proof}

Other definitions could be considered for $\mu$-limit sets, in particular the Cesaro mean could be used.

\begin{definition}[Cesaro-persistent set]
For a CA $\A$, we define the \emph{Cesaro-persistent set} $C_{\mu}(\A)\subseteq Q^*$ by: $\forall u \in Q^*$:
 $$u\notin C_{\mu}(\A) \Longleftrightarrow \lim_{n\rightarrow\infty}\frac{1}{n}\sum_{k\leq n}\A^k\mu([u]_0)=0.$$
 
 Then the \emph{$\mu$-Cesaro-limit set} of $\A$ is $\Lambda C_{\mu}(\A)=\left\{c\in Q^{\Z}:L(c)\subseteq C_{\mu}(\A) \right\}$.
\end{definition}

We then get a lemma equivalent to Lemma~\ref{lem:generic} but for the
Cesaro-persistent set. Its proof is the same. 

\begin{lemma}
\label{lem:cesarogeneric}
Given a CA $\A$, a $\sigma$-invariant measure $\mu\in\Mes{Q^\Z}$ and a finite word $u$, for any weakly generic configuration $c$:\\
\begin{center}$u\notin C_{\mu}(\A)$ $\Longleftrightarrow$  $\lim\limits_{t\to +\infty}\frac{1}{t}\sum\limits_{\tau=0}^{t}d_{\A^{\tau}(c)}(u) = 0$
\end{center}
\end{lemma}

\begin{ex}
  \label{ex:max}
We consider here the ``max'' automaton $\A_M$: the alphabet contains only two states $0$ and $1$. The radius is $1$ and $\delta_{\A_M}(x,y,z)=\max(x,y,z)$.

 The probability to have a $0$ at position $0$ time $t$ is the
 probability to have $0^{2t+1}$ centered on position $0$ in the initial configuration, which tends to $0$ when $t\to \infty$ for the uniform Bernoulli measure, so $0$ does not appear in the $\mu$-limit set. And finally $\Lambda_{\mu}(\A_M)=\{^{\omega}1^{\omega}\}$.

 The limit set of a cellular automaton is defined as
 $\Lambda(\A)=\bigcap_{i\in\N}\A^i(Q^{\Z})$, so
 $\Lambda(\A_M)=(^{\omega}10^*1^{\omega})\cup
 (^{\omega}0^{\omega})\cup (^{\omega}10^{\omega})\cup
 (^{\omega}01^{\omega})$. Actually, we can prove that this limit set
 is an example of limit set that cannot be a $\mu$-limit set \cite{Boyer-Delacourt-Sablik-2010}.

\end{ex}

\begin{ex}
  Consider any CA $\A$ over alphabet $Q$ and add to it a spreading
  state ${s\not\in Q}$: if a cell sees $s$ in its neighborhood it
  becomes $s$, otherwise it behaves according to $\A$. By the same
  reasoning as above, the new CA $\A_s$ obtained this way has a
  trivial $\mu$-limit set (as soon as $\mu$ gives some weight to $s$):
  the singleton made of configuration ${^{\omega}s^{\omega}}$. On the
  other hand, its limit set is as complex as the one from $\A$,
  precisely: its intersection with ${Q^{\Z}}$ is exactly the limit set
  of $\A$.
\end{ex}

In~\cite{bpt-2006}, it is shown that the $\mu$-limit set of the
elementary CA 184 is exactly the pair of configurations
${\bigl\{^{\omega}(01)^{\omega},^{\omega}(10)^{\omega}\bigr\}}$ when
$\mu$ is the uniform measure. It is interresting to note that, on the
contrary, a limit set must be a singleton when it is finite.

Other examples are studied in detail in~\cite{km-2000}.

\section{Construction Toolbox}
\label{sec:constr}

This section is dedicated to the proof of the following theorem.

\begin{theorem}
  \label{thm:maintheorem}
  Given a finite alphabet $Q_0$:
  \begin{enumerate}
  \item for any growing computable sequence $w\in\W(Q_0)$, there
    exists a CA $\A$ over alphabet $Q\supseteq Q_0$ such that, for any
    full-support Bernoulli measure $\mu$ over $Q$, $\Lmu(\A)=L_w$.
  \item for any growing computable sequences $w,w'\in\W(Q_0)$, there
    exists a CA $\A$ over alphabet $Q\supseteq Q_0$ such that, for any
    full-support Bernoulli measure $\mu$ over $Q$, $\left\{\begin{array}{l}\Lmu(\A)= L_w\cup L_{w'}\\ C_{\mu}(\A)=L_{w'}\end{array}\right.$.
  \end{enumerate}
\end{theorem}

It will mainly be used as a tool but it has an immediate corollary
that gives an interesting hint on what is the set of all possible
$\mu$-limit sets (recall that we fixed a global set $\allalphabets$
from which we take any finite alphabet, hence the set of cellular
automata or the set of $\mu$-limit sets is well-defined).

\begin{corollary}
 \begin{align*}
    &\left\{\Lmu(\A),\A\text{ is a CA and $\mu$ the uniform measure
      }\right\}\\
=&\left\{C_\mu(\A),\A\text{ is a
        CA and $\mu$ the uniform measure}\right\}\\
=&\left\{L_{w},w\in\W(Q)\text{ and $Q$ is a finite alphabet}\right\}
  \end{align*}
\end{corollary}
\begin{proof}
  Using Theorem~\ref{thm:maintheorem}, the only part that remains to
  be proven is that any $\Lmu(\A)$ and any $C_\mu(\A)$ is of the form
  $L_w$ for some $w$. Indeed: consider a computable generic configuration $c$ and
  define $w_t$ as the word of size $t$ at the center of
  ${{\A}^t(c)}$. The sequence ${\bigl(d_{w_t}(u)\bigr)_t}$ converges
  towards ${\bigl(d_{{\A}^t(c)}(u)\bigr)_t}$ therefore
  Lemma~\ref{lem:generic} concludes that ${L_w=\Lmu(\A)}$ where
  ${w=(w_t)}$.

  Now if we define $w'_t$ as the concatenation of the words of size
  $t$ at the center of ${{\A}^i(c)}$ for all ${0\leq i\leq t}$ we get
  that the sequence ${\bigl(d_{w'_t}(u)\bigr)_t}$ converges towards
  the sequence 
  \[\left(\frac{1}{t}\sum_{i\leq t}d_{{\A}^i(c)}(u)\right)_t\]
  Lemma~\ref{lem:cesarogeneric} concludes that ${L_{w'}=C_\mu(\A)}$ where
  ${w'=(w'_t)}$.
\end{proof}

Note that the use of the uniform measure is not essential in the above
corollary. The same proof works for any Bernoulli measure with full support
and computable coefficients.

The proof of Theorem~\ref{thm:maintheorem} is constructive and
consists in the description of the CA realizing the desired
$\mu$-limit set. This section will successively deal with the
different parts of the construction after a short overview of the
ideas we use.

\subsection{Overview}
We describe a CA $\A$ over alphabet $Q$ that contains the alphabet
$Q_0$ of the theorem. This CA has $3$ components which work
essentially independently but achieve together the desired
behavior. The general idea is that, starting from a random
configuration, $\A$  will:
\begin{itemize}
\item self-organizes into well-structured computation zones;
\item these zones will evolve with time thanks to a merging process
  that ensures that they grow in size at a controlled rate;
\item inside these well-sized zones a computation process runs
  permanently and essentially fills the zone in an appropriate way with
  the words from the given growing computable sequences.
\end{itemize}

We will describe $\A$ in a incremental way since each stage of the
behavior above makes sense only within some structured zones prepared
by the previous stage. However, all these stages actually run in parallel and
at any time there are zones of the configuration which are completely
out of control. The point is that we will do a reasoning  \emph{in density} starting
from a generic configuration, which justified by
lemmas~\ref{lem:generic} and \ref{lem:cesarogeneric}.

More precisely, the incremental description and analysis of $\A$ will
be the following:

\begin{itemize}
\item Cleaning out the space (Section \ref{sub:cleaning}, alphabet $Q_1$):
  \emph{the density of reliable cells goes to 1.}
  
\item Centralization (Section \ref{sub:centralization}, alphabet $Q_2$ containing a quiescent state $0_2$):
  \emph{the density of cells belonging to a well-sized computation zone
  (size depending on time) goes to 1.}
  
\item Computing and writing (Section \ref{sub:comp_wri}, alphabet $Q_3$ containing a quiescent state $0_3$):
  \emph{within a well-sized computation zone, the content is filled with
  copies of $w_i$ or $w'_i$ (depending on time) up to some set of
  cells whose density goes to $0$.}
\end{itemize}

The alphabet of $\A$ is $Q=Q_0\sqcup\left(\left(Q_1\sqcup (Q_2\times Q_3)\right)\times Q_0\right)$. This can be interpreted the following way:
\begin{itemize}
\item every cell contains a layer filled with some state in $Q_0$, this is the primary layer;
\item some cells contain additionnally another layer (secondary layer) containing either a state in $Q_1$ or in $Q_2\times Q_3$.
\end{itemize}

\subsection{Cleaning out the Space}
\label{sub:cleaning}

\newcommand{\sti}{\vbox to 7pt{\hbox{\includegraphics{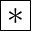}}}\hspace{.1em}}
\newcommand{\sts}{\vbox to 7pt{\hbox{\includegraphics{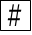}}}\hspace{.1em}}
\newcommand{\std}{\vbox to 7pt{\hbox{\includegraphics{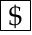}}}\hspace{.1em}}
\newcommand{\stdz}{\vbox to 7pt{\hbox{\includegraphics{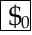}}}\hspace{.1em}}
\newcommand{\stdu}{\vbox to 7pt{\hbox{\includegraphics{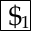}}}\hspace{.1em}}
\newcommand{\stdd}{\vbox to 7pt{\hbox{\includegraphics{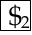}}}\hspace{.1em}}
\newcommand{\ste}{\vbox to 7pt{\hbox{\includegraphics{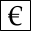}}}\hspace{.1em}}
\newcommand{\stw}{\vbox to 7pt{\hbox{\includegraphics{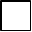}}}\hspace{.1em}}

In this section, we describe the initialization of the construction. Only the secondary layer is concerned and unless stated otherwise, all the states mentionned are in alphabet $Q_1$.   We want to build  a ``protected'' area in a cone of the space-time diagram (the area between two signals moving in opposite directions) and make sure that nothing from the outside can affect the inside of the cone.

\subsubsection{General Description}

The idea is to use a special state $\sti\in Q_1$ that can only appear in the initial configuration (no transition rule produces this state). This state will produce a cone in which a construction will take place. On both sides of the cone, there will be unary counters that count the ``age of the cone''.

The counters act as protective walls to prevent the exterior from affecting the construction. Any information, apart from another counter, is erased. This is a key point, since we will only be interested in well formed structures, that is two walls moving away one from the other and delimiting a totally controled space.  If two counters collide, they are compared and the youngest has priority (it erases the older one and what comes next). Because the construction is assumed to be generated by a state $\sti$ on the initial configuration, no counter can be younger since all other counters were already present on the initial configuration.

The only special case is when two counters of the same age collide. In
this case they both disappear and a special delimiter state $\sts$ is written.

\subsubsection{The Younger, the Better}

The $\sti$ state produces 4 distinct signals. Two of them move towards
the left at speed $1/4$ and $1/5$ respectively. The other two move
symmetrically to the right at speed $1/4$ and $1/5$.

Each pair of signals (moving in the same direction) can be seen as a
unary counter where the value is mostly encoded in the distance between the
two of them, this will be discussed later. As time goes by the signals move apart.

Note that signals moving in the same direction (a fast one and a slow
one) are not allowed to cross. If such a collision happens, the slower
signal is erased. A collision cannot happen between signals generated
from a single $\sti$ state but could happen with signals that were
already present on the initial configuration. Collisions between
counters moving in opposite directions will be explained later as
their careful handling is the key to our construction.

Because the $\sti$ state cannot appear elsewhere than on the initial
configuration and counter signals can only be generated by the $\sti$
state (or be already present on the initial configuration), a counter
generated by a $\sti$ state is at all times the smallest possible one:
no two counter signals can be closer than those that were generated
together. Using this property, we can encapsulate our construction
between the smallest possible counters. We will therefore be able to
protect it from external perturbations: if something that is not
encapsulated between counters collides with a counter, it is
erased. And when two counters collide we will give priority to the
youngest one.

\subsubsection{Dealing with collisions}

Collisions of signals are handled in the following way:
\begin{itemize}
\item an outer signal carries a bit: it can be \emph{open} or \emph{closed};
\item nothing other than an outer signal can go through a \emph{closed} outer
  signal (in particular, no ``naked information'' not contained
  between counters);
\item when two \emph{outer} signals collide they move through each
other, both become \emph{open}, and comparison signals are generated as illustrated by Figure~\ref{fig:speeds}:
\begin{itemize}
\item on each side, a signal S1 moves at maximal speed towards the
\emph{inner} border of the counter, bounces on it ($C$ and $C'$) and
goes back as S2 to the point of collision ($D$);
\item the first signal S2 to come back is the one from the youngest
counter and it then moves back to the \emph{outer} side of the oldest
counter ($E$) and deletes it;
\item the comparison signal from the older counter that arrives
afterwards ($D'$) is deleted and will not delete the younger counter's
\emph{outer} border;
\item all of the comparison signals delete all information that they
encounter other than the two types of borders of counters.
\end{itemize}
\item nothing else than an outer signal or a S2 can go through an
  \emph{open} outer signal;
\item when a signal S2 goes through an \emph{open} outer signal, this
  one becomes \emph{closed};
\item nothing else than an inner signal or an outer signal can stop a
  signal S1; so an S1 signal, either encounters an inner signal and
  bounce on it and becomes S2, or it is destrotyed by another outer
  signal.
\end{itemize}

\begin{figure}[htbp] \centering
  \includegraphics[height=6cm]{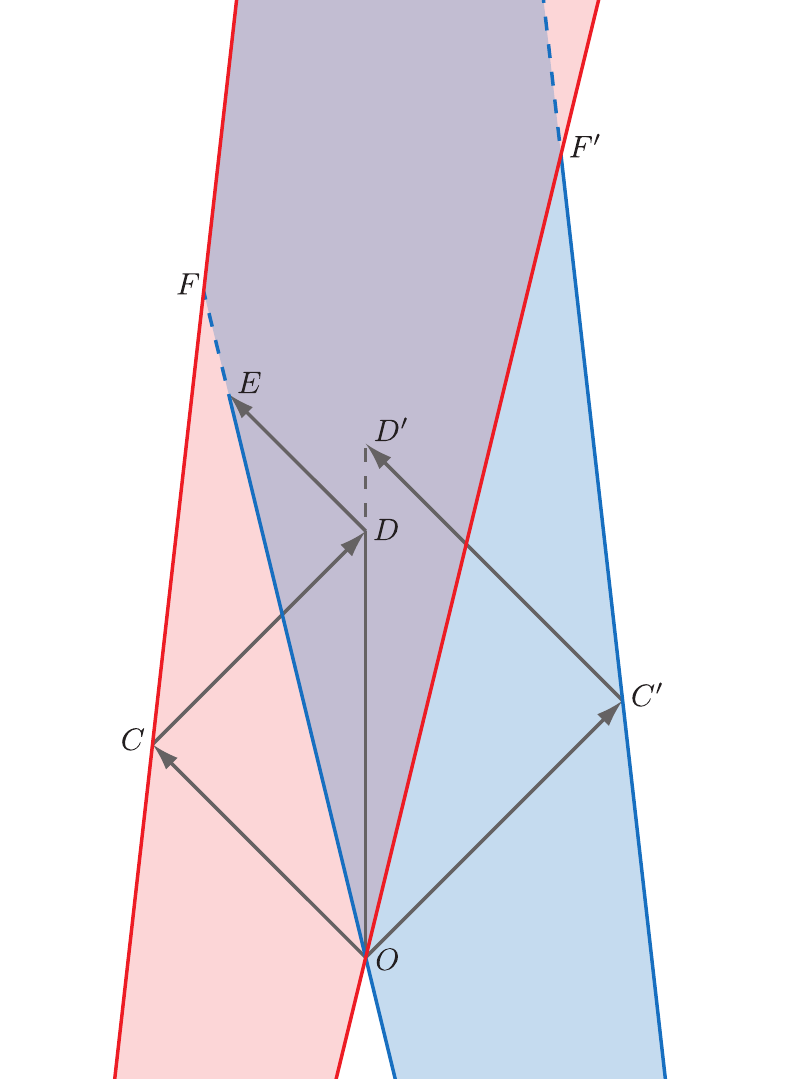}
  \caption{The bouncing signal must arrive (point $E$) before the
older counter moves through the younger one (point $F$).}
  \label{fig:speeds}
\end{figure}

\subsubsubsection{Counter Speeds} It is important to ensure that the
older counter's \emph{outer} border is deleted before it crosses the
younger's \emph{inner} border. This depends on the speeds $s_o$ and
$s_i$ of the \emph{outer} and \emph{inner} borders. It is true
whenever $s_o \geq \frac{1-s_i}{s_i+3}$. If the maximal speed is $1$
(neighborhood of radius $1$), it can only be satisfied if
\[ s_i<\sqrt{5}-2\simeq 0.2360
\] This means that with a neighborhood of radius 1 the \emph{inner}
border of the counter cannot move at a speed greater than $(\sqrt
5-2)$. Any rational value lower than this is acceptable. For
simplicity reasons we will consider $1/5$ (and the corresponding $1/4$
for the \emph{outer} border of the counter). If we use a neighborhood
of radius $k$, the counter speeds can be increased to $k/5$ and $k/4$.

\subsubsubsection{Exact Location}

Note that a precise comparison of the counters is a bit more complex
than what has just been described. Because we are working on a
discrete space, a signal moving at a non integer speed does not
actually move at each step. In particular, in the case of radius $1$, it stays on one cell for a few
steps before advancing, but this requires multiple states.

In such a case, the cell of the signal is not the only
significant information. We also need to consider the current state of
the signal: for a signal moving at speed $1/n$, each of the $n$ states
represents an advancement of $1/n$, meaning that if a signal is
located on a cell $i$, depending on the current state we would
consider it to be exactly at the position $i$, or $(i+1/n)$, or
$(i+2/n)$, etc. By doing so we can have signals at rational
non-integer positions, and hence consider that the signal really moves
at each step.

When comparing counters, we will therefore have to remember both
states of the faster signals that collide (this information is carried
by the vertical signal) and the exact state in which the slower signal
was when the maximal-speed signal bounced on it. That way we are able
to precisely compare two counters: equality occurs only when both
counters are exactly synchronized.

\subsubsubsection{The \emph{Almost} Impregnable Fortress}

Let us now consider a cone that was produced from a $\sti$ state on
the initial configuration. As it was said earlier, no counter can be
younger that the ones on each side of this cone. There might be other
counters of exactly the same age, but then these were also produced
from a $\sti$ state and we will consider this case later (it is the
useful case for the other parts of the construction).

Nothing can enter this cone if it is not preceded by an \emph{outer}
border of a counter. If an opposite \emph{outer} border collides with
our considered cone, comparison signals are generated. Because
comparison signals erase all information but the counter borders, we
know that the comparison will be performed correctly and we do not
need to worry about interfering states. Since the borders of the cone
are the youngest possible signals, the comparison will make them
survive and the other counter will be deleted.

Note that two consecutive opposite \emph{outer} borders, without any
\emph{inner} border in between, are not a problem. The comparison is
performed in the same way. Because the comparison signals cannot
distinguish between two collision points (the vertical signal from $O$
to $D$ in Figure~\ref{fig:speeds}) they will bounce on the first they
encounter. This means that if two consecutive \emph{outer} borders
collide with our cone, the comparisons will be made ``incorrectly''
but this error will favor the well formed counter (the one that has an
\emph{outer} and an \emph{inner} border) so it is not a problem to us.

\subsubsubsection{Evil Twins}

The last case we have to consider now is that of a collision between
two counters of exactly the same age. Because the only counters that
matters to us are those produced from the $\sti$ state, the case we have to consider is the one
where two cones produced from a $\sti$ state on the initial
configuration collide.

According to the rules that were descibed earlier, both colliding
counters are deleted. This means that the right side of the leftmost
cone and the left part of the rightmost cone are now ``unprotected''
and facing each other. A delimiter state $\sts\in Q_2$ (this is the only state outside of $Q_1$ that we consider in this section) is then written and remains where the collision happened,  as illustrated in Figure~\ref{fig:asynccoun}.

\begin{figure}[htbp] \centering
  \includegraphics[width=\textwidth]{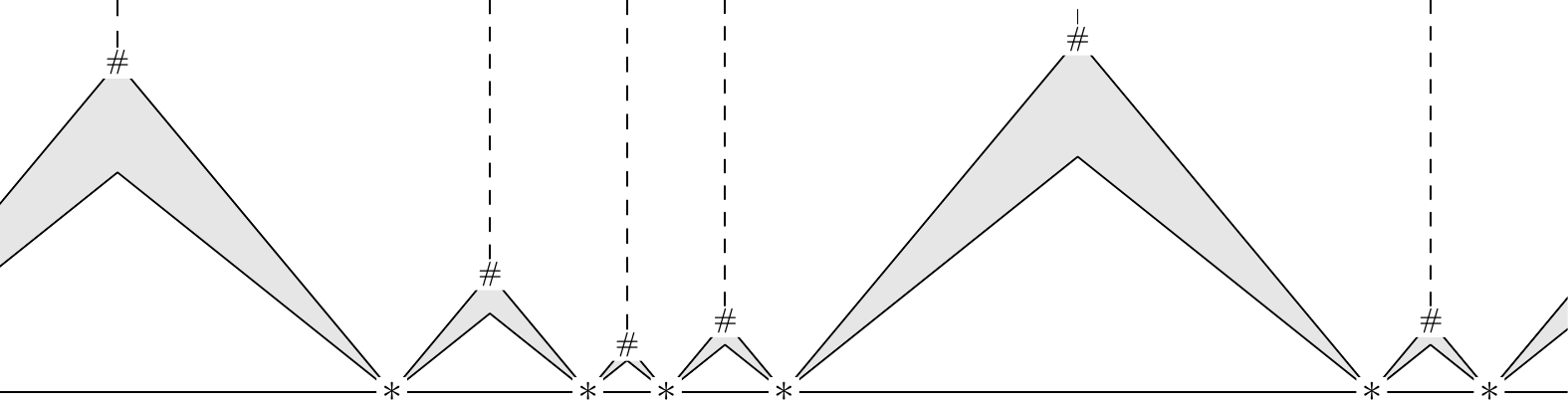}
  \caption{Collision of counters of same age.}
  \label{fig:asynccoun}
\end{figure}

\begin{definition}
  Given any initial configuration $c$, a cell $z$ at time $t$ is
  \emph{reliable} if it is inside the inner cone of some state
  $\sti$ in the initial configuration.
\end{definition}

Note that this definition is independent of the parts of the CA which
are not yet described. Therefore we can already prove a density result
about reliable cells.

\begin{lemma}
  \label{lem:reliable}
  For any non-trivial Bernoulli measure $\mu$ and generic
  configuration $c$, the density of reliable cells goes to $1$ as
  time increases.
\end{lemma}
\begin{proof}
  For a cell at time $t$ and position $z$ to be unreliable it is
  necessary that state $\sti$ does not occur in the initial
  configuration within the interval ${[z-\lambda(t);z+\lambda(t)]}$
  (where $\lambda$ is a linear function related to the slope of the
  inner cone). The configuration $c$ being generic for a non-trivial Bernoulli measure, we immediately have that the density is upper bounded by
  ${(1-p)^{2\lambda(t)}}$ where ${0<p<1}$ is the probability of state
  $\sti$. The lemma follows.
\end{proof}

In the following we will focus on reliable parts of the configuration
only.

\subsection{Centralization}
\label{sub:centralization}

\begin{definition}
  In a configuration $c$ at time $t$, a segment is an interval
  ${[z_1;z_2]}$ of reliable cells not containing $\sts$ and such that both ${z_1-1}$ and ${z_2+1}$ are
  in state $\sts$.
\end{definition}

  In this section, we describe the external behavior of segments, that is how they interact. In the next section, we will need  to dispose of arbitrarily large segments, and to get rid of the small ones. Thus, the idea is to erase some delimiters $\sts$ in order to pool the available space of many segments into a single one.  We will make sure that most segments eventually merge with another one, which means most segments become arbitrarily large through time. We still do not concern ourselves with the primary layer. For the secondary layer, all the states we will use in this section are in the alphabet $Q_2\times Q_3$ that already contains $\{\sts\}\times Q_3$. We will use the alphabet $Q_3$ later, hence suppose that the $Q_3$ component is always $0_3$ (the quiescent state for $Q_3$) for the moment.
 

Now, let us describe the dynamics of segments among themselves. We will specify particular times when merging can happen, independently from the computation performed inside each segment. We will fix a lower bound on the acceptable size of a segment, and at these specific times, any segment that is smaller than this bound will merge. For this purpose we need to synchronize all the segments. As counters compute the time since the initial configuration, we will keep this information in segments. Therefore, time since the initial configuration is an information shared by every segment. With such a protocol, mergings are many to one and not only two to one. 

\subsubsection{Synchronization}

  When a $\sts$ is created by the collision of two counters, their common value of time is written in base $K$, for some $K\geq 2$ (the value will be precised later), on each side of the $\sts$. Hence, the age of each segment is written on both its sides. And every such $K$-ary counter keeps computing time. As any segment is delimited by acceptable $\sts$, this age is the same for all of them and is stored within $\lceil \log_K(t)\rceil$ cells on each side. 

Denote $t_i=\lceil K^{\sqrt{i}}\rceil$ for all $i\in \N$. We allow segments to merge only at time $t_i$ for any $i\in \N$.  We say that a segment is admissible at time $t$ if its length $n$ is such that $\lceil \log_K(t)\rceil \leq \lfloor \sqrt{n} \rfloor$. For any $i$, $\lceil \log_K(t)\rceil = \sqrt{i+1}$ remains unchanged between $t=t_i$ and $t=t_{i+1}$, hence each segment has to decide before $t=t_{i+1}-1$  if  $i+1 \leq n$. If not, the segment decides to merge.
  
To test this condition, segments will measure their own length. This is achieved by sending a signal from the left delimiter to the right one and back. The signal will count the length $n$ in base $K$, then $\lfloor i+1 \rfloor$ is computed inside the $\lceil \log_K(\sqrt{n})\rceil$ leftmost cells. Now each segment knows its age and its size.

\subsubsection{Merging}
 For some $i\in\N$, each segment has to decide whether it will need to merge at time $t_i$ (if it is smaller than $i+1$). If so, it checks whether its neighbors want to merge too. Then the rules to choose which neighbor it will merge with, are the following:
\begin{itemize}
\item if none of its neighbors wants to merge, it merges with the left one,
\item if only one among its neighbors wants to merge, it merges with that one.
\end{itemize}

 Then each $\sts$ delimiter between a segment and the segment it wants
 to merge with is erased. New segments are created between the
 remaining $\sts$. 



\begin{remark}
To prepare itself, a segment that needs to merge before $t_{i+1}$ (suppose we are at timestep $t= t_{i}$) has to:
\begin{itemize}
\item compute its length $n$, which needs $2n$ timesteps;
\item compute $i+1$ which takes time polynomial in $\sqrt{i}$ (value given by the length of the word encoding the age);
\item compare both, linear time;
\item check its neighbors: $n$ timesteps (if we suppose they have achieved their own computations).
\end{itemize}

 A segment needs to merge if $n\leq i+1$, each of these steps requires only polynomial time in $i$, and for large enough $i$ (large enough time), this is achieved in less than $(t_{i+1}-t_i)$ timesteps. So each segment that needs to merge has enough time to decide it before the merging step $t_{i+1}$. Other segments declare nothing to their neighbors, meaning they do not want to merge.

 So mergings can concern:
\begin{itemize}
\item either many segments that all want to merge,
\item or one that wants to merge and one that does not.
\end{itemize}
\end{remark}

\begin{remark}
\begin{enumerate}
\item For any $i\in\N$, after time $t_i$, each segment is larger than $i$.
\item If two segments exactly merge at time $t_i, i \in\N$, at least one of them is smaller than $i$.
\item If three or more segments merge together at time $t_i, i\in\N$, they are all smaller than $i$.
\end{enumerate}
\end{remark}

\subsubsection{$\mu$-limit sets}

The purpose of this slow merging process is to control the size of
segments so that the computation process we put inside can do its job
correctly:
\begin{itemize}
\item we need larger and larger segments to do longer and longer computations,
\item but but we want the typical size to grow slowly enough so that
  the computation process has time to fill the segment with the result
  of the computation. 
\end{itemize}

We already saw that, by construction, segments at time $t_i$ are of
size at least $i$. So it remains to put an upper bound on the desired size of segments.

\begin{definition}
For all $i\in\N$, a segment is said to be \emph{well-sized} at time
$t_i\leq t<t_{i+1}$ if its size is greater than $i$ and less than ${K_i=i^3}$. 
\end{definition}

The main goal of this section is to show that the density of cells
that are inside a well-sized segment goes to $1$. To show this we will
focus on all cells that are in an undesirable situation, \textit{i.e.}
in one of the following cases:
\begin{enumerate}
\item unreliable,
\item reliable but not inside a segment.
\item inside a too large segment (too small segments can not exist by construction),
\end{enumerate}

The first case was solved in the previous section
(Lemma~\ref{lem:reliable}). We will now formalize the others cases. In
the sequel, all the reasoning is done starting from a configuration
$c$ which is generic for some non-trivial Bernoulli measure. The
general idea is to consider (at any time $t$) \emph{maximal
  intervals} of cells which are reliable but not in state $\sts$. For instance, a segment is
such an interval with a $\sts$ at both ends. To each interval we
associate its corresponding pattern in the initial
configuration. Then, to show that a certain kind of interval has a small
density, it is sufficient to show that the corresponding pattern in
the initial configuration has a sufficiently small probability
compared to the length of the interval.

\begin{lemma}
  \label{lem:notsegment}
  Let $d_0(t)$ (resp. $d_0(t,l)$) be the density of reliable cells at time $t$ that are in a
  maximal reliable interval (resp. of size $l$) but not inside a segment. We
  have the following:
  \begin{enumerate}
  \item  ${\exists\alpha, 0<\alpha<1,}$ such that $d_0(t,l)\leq \alpha^l$ for
    any $t$ and for $l$ large enough;
  \item $d_0(t)\xrightarrow[t\to+\infty]{} 0$
  \end{enumerate}
\end{lemma}
\begin{proof}
  Consider a maximal reliable interval $I\subseteq\Z$ which is not a segment at time
  $t$: at least one of its ends corresponds to the inner cone
  generated by some state $\sti$ at position ${z_0\in I}$ in the initial
  configuration. So there is some constant $\alpha>0$ (slope of the
  cone) such that ${|I|\geq \alpha\cdot t}$. Now, considering the
  history of $I$, i.e. the successive maximal reliable intervals $I_{t'}$ containing $z_0$ at any time $t'$
  between $0$ and $t$, we have the following:
  \begin{itemize}
  \item $I_{t'}$ is never a segment;
  \item therefore, if a merging happens at some extremity of some
    $I_{t'}$, it can only be when $t'=t_i$ (for some $i$) and involve
    $I_{t'}$ plus a single segment of size at most $i$.
  \end{itemize}
  We deduce that ${|I|\leq 2\cdot\alpha\cdot t}$ (the worst case being
  when $I$ has not been involved in any merging in all its
  history). Therefore, there must be an interval ${I_0\subset I}$ of
  cells which are not in state $\sti$ in the initial configuration and
  whose size verifies: ${|I_0|\geq |I|/2}$. 

  Denoting by $d_L(t,l)$ the density of cells which are the leftmost
  cell of a maximal reliable interval of size $l$ which is not a
  segment at time $t$, we have:
  \begin{enumerate}
  \item $d_L(t,l)=0$ if $l\leq \alpha\cdot t$;
  \item $d_L(t,l)\leq l\cdot(1-p)^{l/2}$ otherwise,
  \end{enumerate}
  where $p$ is the probability of state $\sti$. Hence $d_0(t,l)\leq
  l^2\cdot(1-p)^{l/2}$ for any $t$ and the first item is proved.
  Moreover we deduce that the density $d_0(t)$ of reliable cells not
  inside a segment at time $t$ verifies:
  \begin{align*}
    d_0(t) &\leq \sum_{l=\alpha t}^{+\infty}l^2\cdot(1-p)^{l/2}\\
    &\leq\sum_{l=\alpha t}^{+\infty}\bigl((1-p)^{\frac{1}{4}}\bigr)^l\text{
      for $t$ large enough}\\
    &=\frac{(1-p)^{\alpha t/4}}{1-(1-p)^{1/4}}
    \xrightarrow[t\to+\infty]{} 0.
  \end{align*}
\end{proof}

\begin{lemma}
  The density $d_1(t)$ of cells which are at time $t$ in a segment which
  is not well-sized goes to $0$ as ${t\rightarrow\infty}$.
  \label{lem:toobigseg}
\end{lemma}
\begin{proof}
  Let $d_+(i)$ be the density of cells at time $t_i$ which are in a
  segment larger than $K_i$ and consider a time $t$ with ${t_i\leq
    t\leq t_{i+1}}$.  Too small segments (\emph{i.e.} smaller than
  $i$) can not exist by construction and too large segments can only
  come from too large segments at time $t_i$ or reliable intervals
  which turned into segments at time $t$. So we have:
  \[d_1(t)\leq d_+(i) + d_0(t-1).\]
  Therefore, using Lemma~\ref{lem:notsegment}, it is sufficient to prove
  that $d_+(i)$ goes to $0$ when ${i\rightarrow\infty}$.

  Consider a segment of size $l\geq K_i$ at time $t_i$. It may only come from one
  of the following situations:
  \begin{enumerate}
  \item a segment of size $l$ at time $t_{i-1}$,
  \item the merging of a segment of size $l-i$ and another of size $i$
    at time $t_{i-1}$,
  \item a reliable interval which is not a segment of size at least $l-i$ at some time $t$,
    ${t_{i-1}\leq t\leq t_i}$,
  \item the merging of many segments of size $i$ at time $t_i$.
  \end{enumerate}
  Going back in time recursively through cases 1 or 2 until encountering
  case 3 or 4, we deduce that to each segment of size ${l\geq K_i}$ at
  time $t_i$ corresponds either a reliable interval which is not a
  segment of size at least ${l-i^2}$ at some time ${t\leq t_i}$, or a
  merging of many segments of size $j$ at time $t_j$ (for some ${j\leq
    i}$) resulting in a large segment of size at least $l-i^2$ (in both
  cases the reduction of $i^2$ in size is an upper bound on the worst
  case where we lose $i$ at time $t_i$, $i-1$ at time $t_{i-1}$, etc).
  Now, denoting by $d_S(i,l)$ the density of cells which are in a
  segment of size $l$ at time $t_i$, and $d_M(j,k)$ the density of
  cells which are at time $t_j$ in a segment of size $k$ resulting
  from a merging of many segments of size $j$, we have:
  \begin{equation}
    d_S(i,l)\leq \frac{l}{l-i^2} \bigl(\sum_{j\leq i} d_M(j,l-i^2) + \sum_{t\leq t_i}d_0(t,l-i^2)\bigr)\label{eq:twocases}
  \end{equation}
  Let's focus first on $d_M(j,k)$ and consider a segment at time $t_j$
  coming from a merge of ${k/j}$ segments of size $j$. Let's
  $z_0\in\Z$ be the leftmost position of that segment: by hypothesis,
  this implies that at each position ${z_0+nj}$, for ${0\leq n\leq
    k/j}$, a $\sts$ is created at some time before $t_j$. This in
  particular implies that the specific pattern ${P=\sti\, a\,\sti\,
  a\,\sti}$ (where ${a\not=\sti}$) does not occur centered at any of
  the aforementioned positions ${z_0+nj}$ in the initial
  configuration: indeed, it would create a $\sts$ at positions
  ${z_0+nj-1}$ and ${z_0+nj+1}$ at time $1$, and forbid forever the
  apparition of the required $\sts$ at position ${z_0+nj}$. Let $q$ be
  the probability of this pattern $P$, we deduce that the density of such
  positions as $z_0$ is less than
  \[(1-q)^{k/j}\]
  and therefore
  \[d_M(j,k)\leq k\cdot(1-q)^{k/j}.\]
  Now putting back this upper bound on $d_M(j,k)$ and the one from
  Lemma~\ref{lem:notsegment} on ${d_0(t,l)}$ into
  Equation~\ref{eq:twocases}, we get
  \[d_S(i,l)\leq \frac{l}{l-i^2} \bigl(i\cdot l\cdot \beta^{\frac{l-i^2}{i}} +
  t_i\cdot\alpha^{l-i^2}\bigr)\]
  where $\alpha$ and $\beta$ are constants between $0$ and $1$. We
  therefore have
  \[d_+(i)\leq \sum_{l\geq K_i} \frac{l}{l-i^2} \bigl(i\cdot l\cdot \beta^{\frac{l-i^2}{i}} +
  t_i\cdot\alpha^{l-i^2}\bigr)\]
  Since $t_i$ grows like ${K^{\sqrt{i}}}$ and $K_i$ grows like ${i^3}$, there is some constant
  $\gamma$ between $0$ and $1$ such that, for large enough $i$, each
  term of the sum above is less than ${\gamma^l}$ (for all ${l\geq
    K_i}$). The lemma follows because for large enough $i$ we then have:
  \[d_+(i)\leq\sum_{l\geq K_i}\gamma^l = \frac{\gamma^{K_i}}{1-\gamma}\xrightarrow[i\to+\infty]{} 0.\]
\end{proof}

From the two lemmas above, we deduce the main result of this section.

\begin{prop}
  \label{prop:outwellseg}
  The density of cells which are inside a well-sized segment at time
  $t$ goes to $1$ as ${t\rightarrow\infty}$.
\end{prop}
\begin{proof}
  The density of such cells is exactly
  ${1-d_{NR}(t)-d_0(t)-d_1(t)}$ where $d_{NR}(t)$ denotes the density
  of non-reliable cells at time $t$.  Lemmas~\ref{lem:reliable},
  \ref{lem:notsegment} and \ref{lem:toobigseg} concludes the proof.
\end{proof}

\subsection{Computing and Writing}
\label{sub:comp_wri}
In this section we describe the final part of the construction. We concentrate on the internal behavior of the segments, and we will use the states of the alphabet $Q_2\times Q_3$. The $Q_2$ part was described in the previous section, and the $Q_3$ part does not interfere with it, hence the state is in $\{O_2\}\times Q_3$ when the $Q_2$ part is not specified by the rules of the previous section.

The computation inside segments depends on two growing
sequences $(w_i)_i$ and $(w'_i)_i$ which can each be generated on
input $i$ in time $\frac{1}{4}(t_{i+1}-t_i)$ and space $\sqrt{i}$.

In a segment at time $t_i, i\in\N$, the computation process goes
through the following steps (see figure~\ref{fig:computationsteps}): 
\begin{enumerate}[label={\bf(\alph{*})}, ref={\bf(\alph{*})}]
\item \label{enum:writ1} in the $\lfloor\sqrt{i}\rfloor$ leftmost cells, the Turing machines $T$ and $T'$ compute and output the words $w_i$ and $w'_i$;
\item \label{enum:writ2} a writing head carrying a memory writes copies of the word $w_i$ separated by some delimiter $\ste$;
\item \label{enum:writ3} the writing head comes back to the left of the segment and wait until $t=t_i+\frac{1}{2}(t_{i+1}-t_i)+(|w_i|+1)K_i$;
\item \label{enum:writ4} a writing head carrying a memory writes copies of the word $w'_i$ separated by some delimiter $\ste$, thus erasing the copies of $w_i$;
\item \label{enum:writ5} the writing head kills itself.
\end{enumerate}

\begin{figure}[ht]
  \centering
  \includegraphics[width=\textwidth]{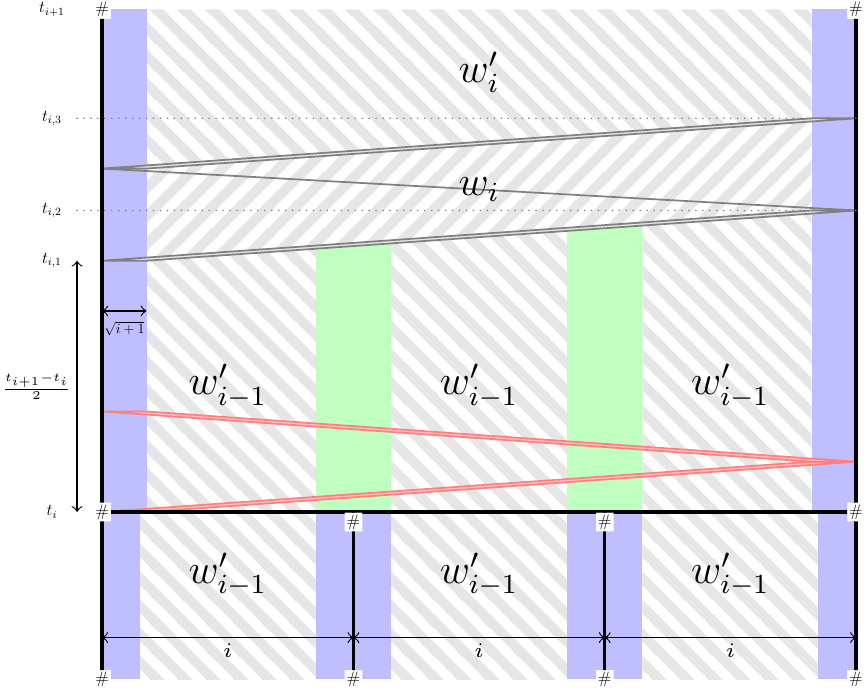}
  \caption{Computation process in a well-sized segment after a merging.}
  \label{fig:computationsteps}
\end{figure}

Suppose the length of the segment is $l$, recall $l\geq i$.

For the first part~\ref{enum:writ1}, the Turing machines are simulated successively in the obvious way, and thanks to the previous remark, they never use more than $\lfloor\sqrt{i}\rfloor$ cells, and the computation is achieved before time $t_i+\frac{1}{2}(t_{i+1}-t_i)$.

For the second part~\ref{enum:writ2}, once the words $w_i$ and $w'_i$ computed and stored in the $\lfloor\sqrt{i}\rfloor$ leftmost cells of the segment, a prefix of the word $(w_i\ste)^{\omega}$ is written all over the segment, this is achieved with a head that carries $w_i$ as its memory, hence it needs $|w_i|.l$ timesteps to reach the end of the segment.

The third part~\ref{enum:writ3} takes only $l$ steps, thus, for any well-sized segment, the head arrives at the left of the segment before time $t_i+\frac{1}{2}(t_{i+1}-t_i)+(|w_i|+1)K_i$. In the case of a non well-sized segment, the writing process stops there.

The fourth part~\ref{enum:writ4} is similar to the third one and the fifth part~\ref{enum:writ5}  is instantaneous.

The whole process takes less than $\frac{1}{4}(t_{i+1}-t_i)+\frac{1}{4}(t_{i+1}-t_i)+(|w_i|+1)K_i+|w'_i|l\leq \frac{1}{2}(t_{i+1}-t_i)+(2\sqrt{i}+1)l$ which is less than $t_{i+1}-t_i$ for any well-sized segment ($l\leq K_i$).

For any segment $s$ at time $t$, denote $w_t(s)$ the content of $s$, that is $F^t(c)_{[a,b]}$ if the segment is between positions $a$ and $b$. Therefore, for any $s$ and $t$, $w_t(s)$ is the concatenation of two subwords: the beginning of the result of the computation in the segment, and the end of the results written by its predecessors, which may contain $\ste$ states. One of those two parts may be empty.

\subsection{Proof of the theorem}
Recall the statement of the theorem:
\begin{theorem}
  Given a finite alphabet $Q_0$:
  \begin{enumerate}
  \item for any growing computable sequence $w\in\W(Q_0)$, there exists a CA $\A$ over alphabet $Q\supseteq Q_0$ such that $\Lmu(\A)=L_w$ where $\mu$ is a full-support Bernoulli measure over $Q$.
  \item for any growing computable sequences $w,w'\in\W(Q_0)$, there exists a CA $\A$ over alphabet $Q\supseteq Q_0$ such that $\left\{\begin{array}{l}\Lmu(\A)= L_w\cup L_{w'}\\ C_{\mu}(\A)=L_{w'}\end{array}\right.$ where $\mu$ is a full-support Bernoulli measure over $Q$.
  \end{enumerate}

\end{theorem}

We will prove the second part of the theorem. We then deduce the first
point by taking $w_i=w'_i,\forall i\in\N$. 
First, by lemma~\ref{lem:incrseq}, remark that it is possible to suppose
that the growing computable sequences $(w_i)_i$ and $(w'_i)_i$ are
such that there exist Turing machines $T$ and $T'$ such that, given
$i\in\N$ as an input, $T$ (resp. $T'$) computes $w_i$ (resp. $w'_i$)
in time $\frac{1}{4}(t_{i+1}-t_i)$ and space $\sqrt{i}$. 

The idea of the proof is that well-sized segments that result from a
merging of well-sized segments, called \emph{good segments}, tend to
almost cover the images of a generic configuration (direct corollary
of Proposition~\ref{prop:outwellseg}), and they contain essentially
copies of the words $w_i$ or $w'_i$. The technical point justifying to
focus only on good segments and not on all well-sized segment is that
a well-sized segment that just merged from a not well-sized segment
might not be properly initialized.

It is essential to note that the content of good segments is easily described as said in the following remark.

\begin{remark}\label{rem:struct_seg}
  For $t_i\leq t\leq t_{i+1}$ and any good segment $s$, $w_t(s)=v_0v_1v_2v_3v_4$ where:
  \begin{itemize}
  \item $|v_0|\leq \sqrt{i}$, this corresponds to the computation area and the storage of the age counter;
  \item $|v_4|\leq \sqrt{i}$, this corresponds to the storage of the age counter on the right;
  \item $|v_2|\leq \sqrt{i}$, this corresponds either to the signal that computes $|s|$ ($t\leq t_i+\frac{1}{2}(t_{i+1}-t_i)$) or to the writing head and its memory ($t\geq t_i+\frac{1}{2}(t_{i+1}-t_i)$);
  \item $v_1$ and $v_3$ belong to $Q_0^*$.
  \end{itemize}
  The words $v_2$ and $v_3$ may be empty.
  Moreover, $v_1$ and $v_3$ contain periodic repetitions of $w'_{i-1}$ (written before $t_i$), $w_{i}$ or $w'_{i}$, depending of the status of the writing process.
  In particular, for $t=t_i$, $v_2$ and $v_3$ are empty and $v_1$ contains repetitions of  $w'_{i-1}$. At $t=t_{i+1}-1$, $v_2$ and $v_3$ are empty and $v_1$ contains repetitions of  $w'_{i}$.
\end{remark}

Denote for $i\in\N$:
\newcommand\tiun{t_{i,1}}
\newcommand\tideux{t_{i,2}}
\newcommand\titrois{t_{i,3}}
\begin{itemize}
\item $\tiun=t_i+\frac{1}{2}(t_{i+1}-t_i)$: at that time, the writing process starts in good segments;
\item $\tideux =t_i+\frac{1}{2}(t_{i+1}-t_i)+(|w_i|)K_i$: at that time, copies of $w_i\ste$ are written all over well-sized segments; 
\item $\titrois =t_i+\frac{1}{2}(t_{i+1}-t_i)+(|w_i|+|w'_i|+1)K_i$: at that time, the writing process is finished.
\end{itemize}

First, the following Lemma justifies that we focus on the density of
words inside good segments.
\begin{lemma}
  \label{lem:segmentsdothejob}
  Take $u\in Q^*$ and $c$ a generic configuration, and consider
  ${(\alpha_t)}$ such that $d_{w_t(s_t)(u)}\geq \alpha_t$ for any
  good segment $s_t$ at time $t$ starting from $c$. We have the
  following:
  \begin{enumerate}
  \item if ${\alpha_t\not\rightarrow 0}$ then ${u\in \Lmu(\A)}$
  \item if ${\alpha_t}$ does not go to $0$ in Cesaro mean then ${u\in C_{\mu}(\A)}$
  \end{enumerate}
\end{lemma}

\begin{proof}
  From Proposition~\ref{prop:outwellseg}, for any $\epsilon>0$ and for
  any large enough $t$, the density of cells outside good segments is
  at most $\epsilon$ (because the density at step $t_{i+1}$ of
  well-sized segments which are not good is less than the density at
  step $t_i$ of segments which are not well-sized). Therefore we have
  from the hypothesis:
  \[d_{{\A}^t(c)}(u)\geq (1-\epsilon)\cdot\alpha_t\]

  Lemma~\ref{lem:generic} and~\ref{lem:cesarogeneric} then allow to
  conclude.
\end{proof}

\begin{claim}
$L_{w'}\subseteq C_{\mu}(\A)$
\end{claim}
\begin{proof}
  Let $u\in L_{w'}$ and $c$ a generic configuration.
  
  There exists  $\epsilon >0$ such that $\forall i_0\in\N, \exists
  i\geq i_0$ such that   $d_{w'_i}(u)>\epsilon$.  We will consider
  $t\in\N$ such that $\titrois\leq t\leq t_{i+1}-1$. For $i$ large
  enough, in every good segment $s_t$ at time $t$,  $v_0$ and
  $v_4$ are negligible in the description of
  Remark~\ref{rem:struct_seg}.
  Therefore \[d_{w_t(s_t)}(u)\geq \frac{1}{2}d_{(w'_i\ste)^{\omega}}(u)\geq \frac{1}{4}d_{w'_i}(u)\geq \frac{1}{4}\epsilon.\]
  So we put ${\alpha_t=\frac{1}{4}\epsilon}$ for $t$ verifying
  $\titrois\leq t\leq t_{i+1}-1$ ($i$ large enough) and ${\alpha_t=0}$ elsewhere.
   Summing at time $t_{i+1}-1$, we get:
  \begin{align*}
  \frac{1}{t_{i+1}-1}\sum\limits_{t=0}^{t_{i+1}-1}\alpha_t\geq & \frac{1}{t_{i+1}-1}\sum\limits_{t=\titrois}^{t_{i+1}-1}\alpha_t\\
  \geq &\frac{1}{t_{i+1}-1}(t_{i+1}-1-\titrois)\frac{1}{4}\epsilon\\
  \geq &\frac{1}{16}\epsilon
  \end{align*}

  We conclude thanks to Lemma~\ref{lem:segmentsdothejob}.
\end{proof}

\begin{claim}
$L_w\subseteq \Lmu(\A)$
\end{claim}

\begin{proof}
  Similarly as in the previous proof, there exists $\epsilon>0$ such
  that for arbitrarily large $i$ $d_{w_i}(u)>\epsilon$. Then we can lower-bound
  the density of $u$ in any good segment at time $\tideux$ by
  ${\frac{1}{4}\epsilon}$, and we conclude thanks to Lemma~\ref{lem:segmentsdothejob}.
\end{proof}

\begin{claim}
$\Lmu(\A)\subseteq L_{w'}\cup L_{w}$
\end{claim}

\begin{proof}
  Take $u\notin  L_{w'}\cup L_{w}$, for any $t\in\N$, the density of
  $u$ in $\A^t(c)$ is due to occurences of $u$ outside good
  segments (density $d_t(u)$) and occurences inside good
  segments (density $d'_t(u)$). Since the density of cells outside
  good segments tends to $0$ (Proposition~\ref{prop:outwellseg}), $d_t(u)\to_{t\to\infty} 0$.

  Inside good segments, the density of $u$ is less than the sum of densities of $u$ in $w_i$, in $w'_i$ and in cells corresponding to $v_0$, $v_2$ or $v_4$ in the description of these segments made in Remark~\ref{rem:struct_seg}, that is $\forall t_i\leq t <t_{i+1}, d'_t(u)\leq d_{w_i}(u)+d_{w'_i}(u)+\epsilon_i$ with $\epsilon_i\to_{i\to\infty} 0$. This proves the claim, using Lemma~\ref{lem:generic}.
\end{proof}

\begin{claim}
$C_{\mu}(\A)\subseteq L_{w'}$
\end{claim}

\begin{proof}
  As $C_{\mu}(\A)\subseteq \Lmu(\A)$, it is enough to prove that $\forall u\in L_{w}\setminus L_{w'}, u\notin C_{\mu}(\A)$. Take such an $u$ and a generic configuration $c$.

  For $i\in\N$, recall $\tiun=\frac{1}{2}(t_{i+1}-t_i)$ and $\titrois
  =\frac{1}{2}(t_{i+1}-t_i)+(|w_i|+|w'_i|+1)K_i$. For any $t\in\N$
  such that $t_i\leq t\leq \tiun$ or $\titrois\leq t<t_{i+1}$, it is
  possible to bound $d_{\A^{t}(c)}(u)$ by some $\epsilon_i$ with
  $\epsilon_i\to_{i\to\infty} 0$. Indeed, for such $t$, in every good
  segment, in the description of Remark~\ref{rem:struct_seg}, $v_1$
  and $v_3$ contain copies of $w'_i$ or $w'_{i-1}$; and cells outside
  good segment have a density going to zero (from
  Proposition~\ref{prop:outwellseg}). To simplify the proof, let's choose
  $\epsilon_i$ such that it is also a bound on the density of cells
  outside good segments for any $t$ between ${t_i}$ and $t_{i+1}$.

  Then for every good segment at time $t_i$:
  \[
  \sum\limits_{t=t_i}^{t_{i+1}-1}d_{w_t{s}}(u)  =
  \sum\limits_{t=t_i}^{\tiun}d_{w_t{s}}(u)+\sum\limits_{t=\tiun}^{\titrois}d_{w_t{s}}(u)+\sum\limits_{t=\titrois}^{t_{i+1}-1}d_{w_t{s}}(u)\]
  We deduce that:
\[
\begin{array}{rcl}
  \sum\limits_{t=t_i}^{t_{i+1}-1}d_{\A^t(c)}(u) &  \leq  & \frac{1}{2}(t_{i+1}-t_i)\epsilon_i+(|w_i|+|w'_i|+1)K_i+\frac{1}{2}(t_{i+1}-t_i)\epsilon_i+(t_{i+1}-t_i)\epsilon_i\\
  \sum\limits_{t=t_i}^{t_{i+1}-1}d_{\A^t(c)}(u) &  \leq  & 2(t_{i+1}-t_i)\epsilon_i + (2\sqrt{i}+1)K_i
  \end{array}\]

  Now take $t\in\N$, there exists $i\in\N$ such that $t_i\leq t< t_{i+1}$, hence
  \[\begin{array}{rcl}
  \frac{1}{t}\sum\limits_{\tau=0}^{t}d_{\A^{\tau}(c)}(u) &  \leq &  \frac{1}{t}\sum\limits_{j=0}^{i}\sum\limits_{\tau=t_j}^{t_{j+1}-1}d_{\A^{\tau}(c)}(u)\\
  &  \leq & \frac{1}{t}\sum\limits_{j=0}^{i} \left(2(t_{j+1}-t_j)\epsilon_j + (2\sqrt{j}+1)K_j\right)\to_{t\to\infty} 0
  \end{array}\]

\end{proof}

\section{Building complex $\mu$-limit sets}
\label{sec:complex}
\subsection{Complexity upper-bounds}
\label{sec:complexupper}
Before giving examples of complex $\mu$-limit sets, let's establish some upper bounds.

A word $w$ is a \emph{wall} for a CA $F$ if for any $c,c'\in[w]_0$ we have:
\begin{enumerate}
\item if $c_z=c'_z$ for every $z<0$ then $F^t(c)_z=F^t(c')_z$ for every $z<0$ and any $t\geq 1$
\item if $c_z=c'_z$ for every $z\geq|w|$ then $F^t(c)_z=F^t(c')_z$ for every $z\geq|w|$ and any $t\geq 1$
\end{enumerate}

It is well-known that a one-dimensional CA $F$ has equicontinuous points if and only if it has walls \cite{Kurka-1997}.

The following proposition is a generalization of theorem 1 of \cite{bpt-2006} to a broader class of measures.
\begin{prop}
  \label{prop:ergowall}
  Let $\mu$ be a $\sigma$-ergodic measure with full support and $F$ a CA admitting $w$ as a wall. Then $L_\mu(F)$ is exactly the set of words occuring in the (temporal) period of the orbit of some (spatially) periodic configuration of period $wu$ for some $u$, formally:
  \[v\in L_\mu(F) \iff \exists t,p\geq 1,v_1,v_2,u\text{ such that }\left\{
  \begin{array}{l}
    F^t\bigl({}^\omega (wu)^\omega\bigr)={}^\omega (v_1vv_2)^\omega\text{ and,}\\F^p\bigl({}^\omega (v_1vv_2)^\omega\bigr)={}^\omega (v_1vv_2)^\omega
  \end{array}\right.
\]
\end{prop}
\begin{proof}
  First, consider some word $v$ occuring in the period of the orbit of ${}^\omega(wu)^\omega$ as in the proposition. Then, for each $k\geq 0$, we have $[wuw]\subseteq F^{-t-kp}\bigl([v_1vv_2])$ because $w$ is a wall for $F$. Hence $F^{t+kp}\mu([v])\geq \mu([wuw])>0$ because $\mu$ has full support, which shows $v\in L_\mu(F)$.

  Suppose now that $v\in L_\mu(F)$. By definition there is $\epsilon>0$ and a sequence $(t_n)$ such that, for all $n$, ${F^{t_n}\mu([v])\geq\epsilon}$. Consider for any $k\geq 0$ the set:
  \[X_k=\bigcup_{-k\leq i\leq k}[w]_i\]
  The union ${X=\cup_{k\geq 0} X_k}$ has measure $1$ because $\mu$ is $\sigma$-ergodic, $X$ is $\sigma$-invariant, ${[w]_0\subseteq X}$ and $\mu$ has full support. Moreover the sequence $X_k$ is increasing, so there is $k_0$ such that ${\mu(X_{k_0})> 1-\frac{\epsilon}{2}}$. By $\sigma$-invariance of $\mu$ we deduce that the set 
\[Y = \sigma^{k_0+|w|}(X_{k_0})\cap\sigma^{-k_0-|v|-1}(X_{k_0})\]
is such that ${\mu(Y)>1-\epsilon}$. Hence, for any $n$, ${F^{-t_n}([v])\cap Y\not=\emptyset}$. We deduce that there is some sub-sequence $(t_{n_p})$ such that, for some ${i<|w|}$ and ${j>|v|}$, and for any $p$, ${F^{-t_{n_p}}([v])\cap [w]_i\cap [w]_j\not=\emptyset}$ (recall that $\sigma$ is the ``left'' shift). Using the fact that $w$ is a wall, we conclude that $v$ occurs in the (temporal) period of the orbit of some (spatially) periodic configuration of period $wu$ for some $u$.
\end{proof}

\begin{theorem}\label{thm:equica}
  Let $\A$ be any CA and $\mu$ a translation invariant measure. We have the following upper bounds:
  \begin{itemize}
  \item if $\mu$ is computable then $L_\mu(\A)$ is a
    $\Sigma_3^0$ arithmetical set;
  \item if $\mu$ is $\sigma$-ergodic with full support and $\A$ has equicontinuity points, then $L_\mu(\A)$ is recursively
    enumerable.
  \end{itemize}
\end{theorem}
\begin{proof}
  Since $\mu$ is computable by some function ${f:A^\ast\times\Q\rightarrow\Q}$, there is a computable function $g:A^\ast\times\Q\times\N\rightarrow\Q$ such that for any $\epsilon$, any $t\in\N$ and any $u$:
\[\bigl|\A^t\mu([u]_0)-g(u,\epsilon,t)\bigr|\leq\epsilon.\]
Indeed, it is sufficient to compute $\A^{-t}(u)$ and sum $f(v,\epsilon')$ for all elements $v$ of this set and a computably small enough $\epsilon'$.
Then, from the definition of $L_\mu(\A)$ we have
  \[u\not\in L_\mu(\A)\Leftrightarrow \forall\epsilon>0,\exists t_0, \forall t\geq t_0, g(u,\epsilon,t)\leq\epsilon.\]
Therefore $L_\mu(\A)$ is $\Sigma_3^0$.

Now suppose that $\mu$ is $\sigma$-ergodic with full support and that $\A$ has equicontinuous points. By hypothesis $\A$ admits some wall $w$ (see \cite{Kurka-1997}). Therefore Proposition~\ref{prop:ergowall} ensures that  $L_\mu(\A)$ is the set of words occuring in the (temporal) period of the orbit of some (spatially) periodic configuration of period $wu$ for some $u$. Since the temporal cycle reached from a spatially periodic initial configuration is finite and recursively bounded in the size of the spatial period, $L_\mu(\A)$ is recursively enumerable.
\end{proof}

\subsection{$\Sigma_3$-hard example}

Here we will prove that the $\mu$-limit language of a cellular automaton can have complexity $\Sigma_3$-hard. For that, with the help of the construction described in Section~\ref{sec:constr}, we will prove a reduction from a $\Sigma_3$-hard problem on Turing machines.

 \begin{definition}
A Turing machine $M$ is said to be \emph{co-finite} (and we write $M\in COF$) when there exists $i_0\in\N$ such that $M$ halts on  every input $i\geq i_0$. 
\end{definition}

The following result was proved in \cite{Odifreddi-1999}.
\begin{theorem}
The problem $COF$ has complexity $\Sigma_3$-hard.
\end{theorem}

Now we can prove that:
\begin{theorem}\label{thm:sigma3}
There exists a cellular automaton $\A$ such that $\Lmu(\A)$ is $\Sigma_3$-complete for every fully supported Bernoulli measure $\mu$.
\end{theorem}

\begin{proof}
  We already know that this problem is $\Sigma_3$ at most. We will use Theorem~\ref{thm:maintheorem} to prove the completeness. Let us describe the growing computable sequence $(w_i)_{i\in\N}$ that will be used.

  First consider a computable enumeration $f$ of $\N^3$, such that for any $(j,k,l)\in\N^3$ there exist infinitely many $i\in\N$ with $f(i)=(j,k,l)$. (Any such enumeration will do.) Let $(\phi_n)_{n\in\N}$ be a computable enumeration of Turing machines. We describe the Turing machine $T$ such that $T$ outputs $w_i$ when given the input $i$.

 Take $i\in\N$, there are $j,k,l\in\N$ such that $f(i)=(j,k,l)$. The idea is to simulate the computation of the Turing machine $\phi_j$ on some particular sequence of consecutive inputs then choose $w_i$ according to the results of the computations, i.e. depending whether the machine halts on each input in this sequence or not.  We will say that $i$ is successful if the machine does halt on each input. Indeed saying that the machine is co-finite means that there exists $k\in\N$ such that the computation ends on all sequences $\{k,k+1,\dots,k+l\}$. Thus, $w_i$ will be a witness (taking the form of a prefix of $(\stdz \stw^j\stdu \stw^k\stdd)^{\omega}$) of the success of a sequence starting at $k$. We will have to avoid writing a witness for $k$ more than once for each $l$.

More formally, $T$ does the following on input $i$:
\begin{itemize}
\item Compute $(j,k,l)=f(i)$.
\item Compute $i_0=\max\{i'<i,f(i')=(j,k,l)\}$.
\item At the same time, simulate the machine $\phi_j$ with successive inputs $k,k+1,\dots,k+l$. If one of these simulations does not halt, then stop the sequence of simulations after $2^{i}$ steps (the bound is purely arbitrary). In this case, $i$ is said to be failed.
\item If the machine $\phi_j$ does halt on all these inputs before timestep  $2^{i}$, then denote $\tau$ the exact time used for the whole computation.
\item If $\tau\leq 2^{i_0}$, then $i$ is said to be failed again, since in this case, some smaller integer was declared successful with the same sequence.
\item In the remaining case, $i$ is said to be successful and $w_i=(\stdz \stw^j\stdu \stw^k\stdd)^{i}$.
\item If $i$ is failed, $w_i=\stw^{i(j+k+3)}$.
  
\end{itemize}

Consider the word $u_{j}=\stdz \stw^j\stdu$, we will prove that:
$$u_{j}\in L_w \Leftrightarrow  \phi_j\in COF$$

\begin{claim}
$\phi_j\in COF \Rightarrow u_j\in L_w$
\end{claim}
\begin{claimproof}

 First suppose that $\phi_j\in COF$ for some $j\in \N$. In this case, there exists $k\in \N$ such that $\forall l\geq k$, $\phi_j$ halts on input $l$. This means that for any $(j,k,l), l\in\N$, there exists $i\in\N$ such that $f(i)=(j,k,l)$ and $2^{i}$ timesteps are enough to simulate $\phi_j$ on inputs $k,k+1,\dots,k+l$ and verify that it halts in each case. Thus for every triplet $(j,k,l),l\in\N$, there exists a successful $i_l\in\N$ with $f(i_l)=(j,k,l)$.

 Hence, $w_{i_l}=(\stdz \stw^j\stdu \stw^k\stdd)^{i_l}$. For $l$, and thus $i_l$, large enough, the density of the word $u_j$  in every $w_{i_l}$ is larger than $\frac{1}{2}\frac{1}{|j|+|k|}$ which is a constant.

\end{claimproof}

\begin{claim}
$\phi_j\notin COF \Rightarrow u_j\notin L_w$ 
\end{claim}
\begin{claimproof}
Here, with $j$ fixed, if we take an infinite number of successful integers, they necessarily concern unbounded values of the starting point $k$. We will use the fact that the density of $u_j$ in $(\stdz \stw^j\stdu \stw^k\stdd)^{\omega}$ decreases when $k$ increases.

 Suppose $\phi_j\notin COF$ for $j\in\N$.  Take $\epsilon >0$. 
For any $i\in\N$, if $f(i)=(j',k,l)$ with $j\neq j'$, then $d_{w_i}(u_j)=0$.

 There exists $k_0\in\N$, such that $\frac{1}{|j|+|k_0|}<\epsilon$. As $\phi_j$ is not co-finite, there exists $l_0\geq k_0$ such that $\phi_j$ does not halt on input $l_0$. Thus, there are at most ${l_0}^2$ triplets $(j,k,l)\in\N^3$ with $k\leq k_0$ and $l\leq l_0$. 
 There exists $i_0\in\N$ such that for any  $(j,k,l),k\leq k_0, l\leq l_0$:
\begin{itemize}
\item either every $i\in\N$ with $f(i)=(j,k,l)$ is failed;
\item or there exists $i < i_0$ such that $i$ is successful.
\end{itemize}

Now take $i\geq i_0$.
\begin{itemize}
\item If $f(i)=(j',k,l), j'\neq j$ then  we have $d_{w_i}(u_j)=0$.
\item If $f(i)=(j,k,l), k\geq k_0$ then $d_{w_i}(u_j)\leq\frac{1}{|j|+|k_0|}<\epsilon$.
\item If $f(i)=(j,k,l), k\leq k_0$, then $d_{w_i}(u_j)=0$ since $i\geq i_0$.
\end{itemize}
\end{claimproof}

\end{proof}

\subsection{Descriptive complexity}

In this section we will use Theorem~\ref{thm:maintheorem} to construct
cellular automata whose $\mu$-limit sets are constrained to be in a
specific subshift. The following proposition shows that we can build a
$\mu$-limit inside any effective subshift. However, let's recall that
there are very simple effective subshits which can not be the
$\mu$-limit set of some CA as shown in Example~\ref{ex:max}. The
question of what kind of subshift can appear as $\mu$-limit sets has
been specifically adressed in \cite{Boyer-Delacourt-Sablik-2010}.

\begin{prop}
	\label{prop:effective}
	Given a non-empty effective subshift $S$ over an alphabet $Q_0$, there exists a CA whose $\mu$-limit set is included in $S$ for every fully supported Bernoulli measure $\mu$.
\end{prop}

\begin{proof}
	Because the subshift $S$ is effective, it can be characterized by a recursively enumerable set of forbidden words. We will use Theorem~\ref{thm:maintheorem} with a growing computable sequence in which the word $w_i$ does not contain any of the first $i$ forbidden words.

 Let us describe the Turing machine $T$ that computes the sequence $(w_i)_i$.        
\begin{itemize}
	\item On input $i$, $T$ enumerates and stores the first $i$ forbidden words of $S$.
	\item All possible words of length $i$ over $Q_0$ are then enumerated in lexicographical order, and $w_i$ is the first one that does not contain any of the forbidden words previously enumerated (there exists one because the subshift is non-empty).
\end{itemize}

We now apply Theorem~\ref{thm:maintheorem} with the growing computable sequence $(w_i)_i$ hence it is enough to prove that $L_w$ contains only words in $\Sigma^*$ and none of the forbidden words.

Now let us consider a forbidden word $v$ in the recursively enumerable set that characterizes the subshift $S$. It is the $i^{\text{th}}$ word enumerated for some $i\in\N$, hence it does not appear in $w_j, j\geq i$.
\end{proof}

This proposition does not allow to describe the $\mu$-limit set obtained, except if the subshift is minimal. A subshift is said to be \emph{minimal} (\cite{Lind-Marcus-1995}) when it does not contain a proper subshift. Hence the proposition implies that:

\begin{corollary}
Given a non-empty minimal effective subshift $S$, there exists a cellular automaton whose $\mu$-limit set is $S$ for every fully supported Bernoulli measure $\mu$.
\end{corollary}

We will see now how the previous proposition implies the existence of a cellular automaton whose $\mu$-limit set contains only configurations of high Kolmogorov complexity.

\begin{definition}
	Given a recursive function $f:\{0,1\}^*\rightarrow \{0,1\}^*$, the Kolmogorov complexity relative to $f$ of a string $x\in\{0,1\}^*$ is defined as $K_f(x) = \min\{|y|,\ f(y)=x\}$.
\end{definition}

As such, the definition of Kolmogorov complexity depends heavily on the choice of the function $f$ and it is not properly defined for words $x$ such that $\{|y| \st f(y)=x\}$ is empty. However, it can be shown that there exists a recursive function $U$ such that, for any recursive function $f$, there is a constant $c_f\in \N$ such that, for any string $x\in\{0,1\}^*$ such that $K_f(x)$ is defined, we have $K_U(x)\leq K_f(x)+c_f$. This also implies that $K_U(x)$ is properly defined for all $x$. The Kolmogorov complexity of a string $x$ is then defined as $K(x)=K_U(x)$ for some such \emph{additively optimal} $U$.

Informally, the Kolmogorov complexity of a word is the length of a
shortest program which outputs that word.

\begin{definition}[$\alpha$-complexity]
  Given a constant $\alpha > 0$, a word of length $n$ on the alphabet $\{0, 1\}^*$ is said to be \emph{$\alpha$-complex} if its Kolmogorov complexity is greater than $\alpha n$. A word that is not $\alpha$-complex is said to be \emph{$\alpha$-simple}.
\end{definition}

\begin{corollary}[of Proposition~\ref{prop:effective}]
	\label{coro:kolmo}
	For any $\alpha < 1$, there exists a constant $n_\alpha$ and a cellular automaton whose $\mu$-limit set contains only configurations whose factors of length greater than $n_\alpha$ are all $\alpha$-complex for every fully supported Bernoulli measure $\mu$.
\end{corollary}

\begin{proof}
	To use Proposition~\ref{prop:effective} we need to show that for some $n_\alpha$ the subshift of configurations over $\{0, 1\}$ that contain no $\alpha$-simple word of length greater than $n_\alpha$ is effective and non-empty.
	
	As for the effectiveness, a word $x$ is $\alpha$-simple if and only if there exists $y$ such that $U(y) = x$ and $|y| \leq \alpha |x|$. We can enumerate all such words by dovetailing the computations of $U(y)$ for all possible $y$ and checking if the resulting word is $\alpha$-simple by comparing its length to that of the input $y$. Therefore the set of $\alpha$-simple words $\{x,\ K(x)\leq \alpha |x|\}$ is recursively enumerable, and so is the set of such words of length greater than $n_\alpha$.

	The existence of $n_\alpha$ and a configuration containing no $\alpha$-simple factor of length greater than $n_\alpha$ is a consequence of the main result in \cite{RumUsh-2006} since there exist at most $2^{\alpha n}$ forbidden words of length $n$ and complexity less than $\alpha n$.
\end{proof}

\begin{corollary}[of Corollary~\ref{coro:kolmo}]
  There exists a CA whose $\mu$-limit set contains only non-recursive configurations for every fully supported Bernoulli measure $\mu$.
\end{corollary}
\begin{proof}
  In a recursive configuration $c$, the word $c_{[0,n]}$ starting at position $0$ and of length $n$ has complexity ${O\bigl(\log(n)\bigr)}$. Therefore no recursive configuration can be $\alpha$-complex in the sense defined above. Corollary~\ref{coro:kolmo} concludes the proof.
\end{proof}

As a last application of Proposition~\ref{prop:effective}, we will show that the quasi-periodicity of a $\mu$-limit set can be highly non-trivial using a result of~\cite{Ballier-Jaendel-2010}. A configuration $c$ is said \emph{quasi-periodic} if any pattern occurring in $c$ occurs in any large enough pattern of $c$. Any subshift contains a quasi-periodic configuration \cite{birkhoff}. For such configurations the quasi-periodicity can be quantified through the \emph{quasi-periodicity function}.

\begin{definition}
  Let $c$ be a quasi-periodic configuration. We associate to $c$ the \emph{quasi-periodic function} ${\rho_c:\N\rightarrow\N}$ defined by:
  \[\rho_c(n) = \max_{u\in L(c), |u|=n}\min\{p : \text{any pattern of size $p$ of $c$ contains $u$}\}\]
\end{definition}

\begin{corollary}[of Proposition~\ref{prop:effective}]
  \label{coro:quasiper}
  There exists a cellular automaton such that for any quasi-periodic configuration $c$ of its $\mu$-limit set, the function $\rho_c$ can not be bounded by any recursive function for every fully supported Bernoulli measure $\mu$.
\end{corollary}
\begin{proof}
  It is a direct application of Proposition~\ref{prop:effective} with the effective subshift obtained by corollary 3.4 of \cite{Ballier-Jaendel-2010}.
\end{proof}

\section{Complexity of properties of $\mu$-limit sets}
\label{sec:properties}

  \subsection{A Rice theorem for $\mu$-limit sets}
 In the case of the limit set of cellular automata, J. Kari~\cite{Kari-1994} proved a result equivalent to Rice theorem, meaning that any non trivial property of limit sets of cellular automata is undecidable. Using certain aspects of his technique, we will prove here that any non trivial property of $\mu$-limit sets of cellular automata has a higher complexity than the negation of the problem of being co-finite for a Turing machine. Since we will deal with different cellular automata in this section, the considered measures will be the uniform ones on each alphabet.

\subsubsection{Properties of $\mu$-limit sets}

Intuitively, a property of the $\mu$-limit set is a property $\mathcal{P}$ which depends only on the $\mu$-limit set: if two CA have the same $\mu$-limit set, then either both have property $\mathcal{P}$ or none has property $\mathcal{P}$.
We use the same formalism as J. Kari for limit sets. Recall that we
have since the beginning consider a countable set
$\allalphabets=\{q_0,q_1,\dots\}$ from which we take finite subsets to
define alphabets.

\begin{definition}
A property $\mathcal{P}$ of $\mu$-limit sets of cellular automata is a subset of the powerset $\mathscr{P}(\allalphabets^{\Z})$. A $\mu$-limit set of some cellular automaton is said to have property $\mathcal{P}$ if it is included in $\mathcal{P}$.
\end{definition}

For example, $\mu$-nilpotency is given by the family $\{{q_i^{\Z}},i\in\N\}$. We will talk equivalently of properties of $\mu$-limit sets and $\mu$-limit languages, but a property of cellular automata \emph{concerning} the $\mu$-limit set is not necessarily a property \emph{of} $\mu$-limit sets. Surjectivity is the classical example to show that both differ. Indeed surjectivity refers to the set of states of the automaton and not necessarily only to those appearing in the $\mu$-limit set. Note also that there is no obvious relationship between properties of $\mu$-limit sets and properties of limit sets:
\begin{itemize}
\item nilpotency is a property of limit sets but not a $\mu$-limit property (\emph{e.g.} for $\mu$ the uniform Bernoulli measure, any CA with a spreading state is $\mu$-nilpotent but can be nilpotent or not);
\item conversely, $\mu$-nilpotency is a property of $\mu$-limit sets, but it is not known whether it is a property of limit sets.
\end{itemize}

 A property is said to be \emph{trivial} when either it contains all $\mu$-limit sets or none.

\subsubsection{Computing a weakly generic configuration}
In order to prove this Rice theorem, we will need to be able to compute the prefixes of some weakly generic configuration, we will then refer to the following proposition proved in \cite{fk-1977}:
\begin{prop}\label{prop:constrgeneric}
There exists a computable weakly generic configuration $c_{WG}$ on the finite alphabet $X$ such that there exist $A,B>0$ such that for any $l\in\N$, $u\in X^l$ and $L\geq |X|^{2l}$, we have:
  $$A|X|^{-l}\leq d_{{c_{WG}}_{[0,L-1]}}(u)\leq B|X|^{-l}$$
\end{prop}

\begin{remark}\label{rem:imggen}
The property over the densities of prefixes can be extended to images of $c_{WG}$ by a cellular automaton, for $k\in\N$ and $L\geq 2k$:
  $$A/2d_{\A^k(c_{WG})}(u) \leq d_{\A^k(c_{WG})_{[0,L-1]}}(u)\leq 2B d_{\A^k(c_{WG})}(u)$$
\end{remark}

\subsubsection{Construction}

\begin{theorem}\label{thm:rice}
Given a property $\mathcal{P}$ of $\mu$-limit sets, either $\mathcal{P}$ is trivial or $\mathcal{P}$ is $\Pi_3$-hard.
\end{theorem}

To prove this theorem, we will use a reduction to the  problem of being co-finite for a Turing machine  which is $\Sigma_3^0$-complete.

The general idea of the proof is close to what J. Kari did for limit sets, using the following proposition:

\begin{prop}\label{prop:rice}
There is an algorithm that, given a cellular automaton $\A$ and a Turing machine $\phi$, produces a cellular automaton $\B$ such that:
\begin{itemize}
\item if $\phi\in COF$ then $\mul(\B)=Q_{\A}^{\Z}$;
\item else $\mul(\B)=\mul(\A)$.
\end{itemize}
\end{prop}

Using this property, whose proof will follow, we can prove Theorem~\ref{thm:rice}. 

\begin{proof}[Proof of Theorem~\ref{thm:rice}]
  Given some non trivial property $\mathcal{P}$ of $\mu$-limit sets,
  consider a Turing machine $\phi$ and cellular automata $\A_1$ and
  $\A_2$ such that exactly one among $\A_1$ and $\A_2$ has property
  $\mathcal{P}$. We consider they have a common alphabet, which is
  always possible by increasing their alphabets if necessary.  We
  reduce the decision problem $\mathcal{P}$ to $COF$ as follows.
  First denote $\B_1$ and $\B_2$ the cellular automata given by
  Proposition~\ref{prop:rice} for respectively $\phi$ and $\A_1$ and
  $\phi$ and $\A_2$. Then using the oracle for $\mathcal{P}$ on $\B_1$
  and $\B_2$, we can decide if the answer is the same or not. The
  first case corresponds necessarily to $\phi\notin COF$ and the
  second to $\phi\in COF$. So we decided $COF$ on $\phi$. 
\end{proof}

Now we prove Proposition~\ref{prop:rice}.

\begin{proof}[Proof of Proposition~\ref{prop:rice}]

The proof will have similarities with the one of Theorem~\ref{thm:sigma3}. Denote for this proof $c=c_{WG}$.

It mainly relies on Theorem~\ref{thm:maintheorem}. Again, we make a reduction to the problem of being co-finite for a Turing machine which is $\Sigma_3^0$-complete. Let us describe the computable sequence $w=(w_i)_{i\in\N}$ associated to it. First consider a computable enumeration $f$ of $\N^2$, such that for any $(k,l)\in\N^2$ there exist infinitely many $i\in\N$ with $f(i)=(k,l)$. Denote $T$ the Turing machine that produces $w$.

 For $(k,l)=f(i), i\in\N$, the idea is to simulate the computation of $\phi$ on some sequence of consecutive inputs ($\{k,k+1,\dots,k+l\}$) and output different $w_i$'s whether the machine halts on each input in this sequence or not.  We will say that the sequence is successful if the machine does halt on each input and failed in the other case. Indeed saying that the machine is co-finite means that there exists $k\in\N$ such that all sequences $\{k,k+1,\dots,k+l\}$ are successful. Thus, we will write a witness  of the success of a sequence starting at $k$. We will have to avoid writing a witness for $k$ more than once for each $l$.\\

 More formally, $T$ does the following on input $i$:
\begin{itemize}
\item Compute $(k,l)=f(i)$ and $\nu(i)=\lfloor\log{i}\rfloor$.
\item Compute $i_0=\max\{i'<i,f(i')=(k,l)\}$.
\item Simulate the machine $\phi$ with successive inputs $k,k+1,\dots,k+l$. If one of these simulations does not halt, then stop the simulation after $2^{i}$ steps. In this case, $i$ is said to be failed.
\item If the machine $\phi$ does halt on all these inputs before timestep  $2^{i}$, then denote $\tau$  the exact time used for the whole computation.
\item If $\tau\leq 2^{i_0}$, the whole computation necessarily ended for $i_0$ and again $i$ is said to be failed.
\item Compute $u_i=c_{[o..(\nu(i)-1)]}$ and $v_i=\left(\A^{\nu(i)}(c)\right)_{[o..(\nu(i)-1)]}$.
\item If $i$ is failed, define $w_i=v_i^{k+1}$.
\item In the other case, $i$ is said to be successful: define $w_i=u_iv_i^{k}$.
\end{itemize}

\begin{claim}\label{cla:ricecof}
If $\phi\in COF$ then $L_w=Q_{\A}^{\Z}$.
\end{claim}
\begin{claimproof}

  In this case, there exists $k\in \N$ such that $\forall l\geq k$, $\phi$ halts on input $l$. This means that for any $l\in\N$, there exists $i_l\in\N$ such that $f(i_l)=(k,l)$ and $2^{i_l}$ timesteps are enough to simulate $\phi$ on inputs $k,k+1,\dots,k+l$ and verify that it halts in each case. Thus there are infinitely many successful $i$'s with $f(i)=(k,l)$ for some $l\in\N$.

 Take any word $u\in Q_{\A}^{*}$. For any such large enough successful $i$ ($\nu(i)\geq |Q_{\A}|^{2|u|}$), we hence have  $d_{w_i}(u)\geq\frac{A}{2(k+1)|Q_{\A}|^{|u|}}$ (with $A$ from Proposition~\ref{prop:constrgeneric}) which is a constant as $k$ is fixed.

\end{claimproof}

\begin{claim}\label{cla:ricenoncof}
If $\phi\notin COF$ then $L_w=\Lmu(\A)$.
\end{claim}
\begin{claimproof}
Here, if we take an infinite number of successful sequences, they necessarily concern unbounded values of the starting point $k$. We will use the fact that the space covered by prefixes of an image of $c$ decreases when $k$ increases.

 Take $u\notin\Lmu(\A)$ and $\epsilon >0$. 

 There exists $k_0\in\N$, such that $\frac{1}{k_0+1}<\epsilon/2$. As $\phi$ is not co-finite, there exists $l_0\geq k_0$ such that $\phi$ does not halt on input $l_0$. There are less than ${l_0}^2$ pairs $(k,l)\in\N^2$ with $k\leq k_0$ and $k+l\leq l_0$. 
 Denote $i_0$ the smallest integer such that for any $f(i)=(k,l)$ with $i\geq i_0$ and $k\leq k_0$:
\begin{itemize}
\item either $\phi$ does not halt on some input between $k$ and $k+l$;
\item or there exists $i'\leq i_0$ with $f(i')=(k,l)$ such that $\phi$ halts on all these inputs in less than $2^{i'}$ timesteps.
\end{itemize}

Thus, when $f(i)=(k,l)$ with $i\geq i_0$ and $k\leq k_0$,  $i$ is failed and $w_i=(\left(\A^{\nu(i)}(c)\right)_{[o..(\nu(i)-1)]})^{k+1}$.

Take $i_1\geq i_0$ such that $\forall i\geq i_1,d_{\A^{\nu(i)}(c)}(u)<\epsilon/(4B)$, which exists since $u\notin\Lmu(\A)$.

Now take $i\geq i_1$ with $f(i)=(k,l)\in\N^2$.
\begin{itemize}
\item If $k\geq k_0$ then  $d_{w_i}(u)<\frac{1}{k_0+1}+d_{\A^{\nu(i)}(c)_{[0..\nu(i)]}}(u)<\epsilon$.
\item If $k\leq k_0$, then $i$ is failed and $d_{w_i}(u)<d_{\A^{\nu(i)}(c)_{[0..\nu(i)]}}(u)<\epsilon$.
\end{itemize}

Hence, thanks to Remark~\ref{rem:imggen}, we conclude that $L_w\subseteq\Lmu(\A)$.

The other direction is easier: for $u\in\Lmu(\A)$ and $i\in\N$, $d_{w_i}(u)\geq \frac{1}{2}d_{\A^{\nu(i)}(c)_[0..\nu(i)]}(u)$ which does not tend to $0$.

\end{claimproof}

\end{proof}

In the next section, we will deal more specifically with $\mu$-nilpotency. We leave open the question of properties of higher complexity. For example, being a shift of finite type, a sofic shift or containing a weakly generic configuration\dots{} In particular, it is not known whether there exist properties of arbitrarily high complexity.

\subsection{$\mu$-nilpotency }

Recall that a CA is $\mu$-nilpotent if and only if its $\mu$-limit set is a singleton.

\begin{prop}
  Let $\mu$ be a computable measure. The set of $\mu$-nilpotent CA is $\Pi_3^0$.
\end{prop}
\begin{proof}
  If a CA is $\mu$-nilpotent then the only configuration in the $\mu$-limit set is necessarily of the form ${}^\omega q^\omega$ for some state $q$. Hence, $\mu$-nilpotency is equivalent to the following property:
  \[\forall \epsilon>0, \exists t_0, \forall t\geq t_0, \exists q_0,\ F^t\mu(q_0)\geq 1-\epsilon\]
  Since $\mu$ is computable and the number of states of a CA is finite, the predicate ``${\exists q_0, F^t\mu(q_0)\geq 1-\epsilon }$'' (depending on $t$, $F$ and $\epsilon$) is recursive, which concludes the proof.
\end{proof}

The following theorem is a direct consequence of the Rice Theorem (\ref{thm:rice}) proved earlier.

\begin{theorem}
  \label{thm:hardnilpo}
Let $\mu$ be some Bernoulli measure on the fullshift, the problem of being $\mu$-nilpotent for a cellular automaton is $\Pi_3^0$-complete.
\end{theorem}

\begin{prop}\label{prop:equica}
  Let $\mu$ be a $\sigma$-ergodic measure of full support. Then we have:
  \begin{itemize}
  \item the set of $\mu$-nilpotent CA with a persistent state is co-recursively enumerable;
  \item the set of $\mu$-nilpotent CA with equicontinuous points is $\Sigma_2^0$.
  \end{itemize}
\end{prop}
\begin{proof}
  Using Proposition~\ref{prop:ergowall}, not being $\mu$-nilpotent is equivalent to the existence of different words of same size in the temporal period of some spatially periodic configuration containing a wall. For CA with a persistent state, it is sufficient to test with a wall made of $r$ adjacent persistent states ($r$ being the radius). Hence we can recursively enumerate CA with a persistent state and a pair of different words as said above. The first item of the Proposition follows.
  
  For CA with equicontinuous points, the additional difficulty is that we don't know \textit{a priori} which word is a wall. Testing this costs an additional quantifier. Formally, a CA is $\mu$-nilpotent with an equicontinuous point if and only if
  \[\exists w, q_0\bigl(\forall z,t\ R(w,z,t)\ \wedge\ \forall v\ R'(q_0,w,v)\bigr)\]
  where predicates $R$ and $R'$ are recursive and such that:
  \begin{itemize}
  \item $R(w,z,t)$ checks that $w$ is a wall up to time $t$ and position $z$ and $-z$ (see definition in Section~\ref{sec:complexupper})
  \item $R'(q_0,w,v)$ checks that periodic configuration $wv$ converges to the $q_0$-uniform configuration (exponential time bound is enough to check)
  \end{itemize}
  The second item of the Proposition follows.
\end{proof}

 The definition of $\mu$-nilpotency has been chosen analogously as the definition of nilpotency. But in the case of nilpotent CA, we can show that the limit set contains either a unique uniform configuration or an infinite number of distinct configurations. As this property is false for $\mu$-limit sets, a notion of \emph{weak $\mu$-nilpotency} can be defined. The most natural way is to say a CA is weakly $\mu$-nilpotent when its $\mu$-limit set is finite. Still, some refinements can be considered, such as $\mu$-limit sets containing only uniform configurations or the shift-orbit of one unique periodic configuration.

 In terms of complexity, the alphabet being finite, the second definition (only uniform configurations) is equivalent to classical $\mu$-nilpotency. Thanks to Rice Theorem, other ones are at least as complex, but we need other quantifiers to describe the finite $\mu$-limit set.

\section{Types of convergence towards the limit}
\label{sec:convergence}
\subsection{Simple convergence}

By definition words which are not in the $\mu$-limit language are those whose probability goes to zero as time increases. However, this probability does not always converge for words which are in the $\mu$-limit language. As a consequence, contrary to the limit set, the $\mu$-limit set is generally changed when taking iterates of a given CA.

\begin{theorem}\label{theo:fdeux}
  For any fully supported Bernoulli measure $\mu$, there exist $F$ such that $F$ and $F^2$ do not have the same $\mu$-limit set.
\end{theorem}
\begin{proof}
   To construct such an $F$ it is sufficient to use the counter construction from the proof of Theorem~\ref{thm:maintheorem}, \textit{i.e.} the initialization step. We just use the trick of unary counters to build a growing uniform ``protected area'' alternating between two states: all black (odd steps), or all white (even steps). We keep the same collision rule described in the proof of Theorem~\ref{thm:maintheorem}:
  \begin{itemize}
  \item when two areas of different ages collide, the older is destroyed by the younger;
  \item when two areas of same age collide, they simply merge (it is possible since, having the same age, they have the same uniform content).
  \end{itemize}
  We say a cell is \emph{synchronized} at time $t_0$ if for any ${t\geq t_0}$ it is black when $t$ is odd and white when $t$ is even. Then, using a simplified version of Lemma~\ref{lem:reliable}, we can prove the following:
  \begin{center}
    \begin{minipage}{0.8\linewidth}
      \textbf{Claim.} Starting from a generic configuration, the density of cells which are synchronized at time $t$ goes to $1$ when $t$ grows.
    \end{minipage}
  \end{center}
  It follows that $F^2$ is $\mu$-nilpotent whereas the $\mu$-limit set of $F$ contains two configurations: the ``all black'' and the ``all white''.
\end{proof}

 We say that a CA is \emph{simply convergent for $\mu$} if the probability of appearence of a word $u$ converges for any $u$, \textit{i.e.}
\[\forall u\in A^\ast, \exists \alpha\in\mathbb{R}, \forall \epsilon>0, \exists t_0, \forall t\geq t_0 :\ \bigl|\A^t\mu([u]_0) - \alpha\bigr|\leq\epsilon.\]

Examples of simply convergent CA are $\mu$-nilpotent CA. Indeed, the probability of apparence of any word goes to $0$ except for one word of each size for which it necessarily goes to $1$.

If $F$ is simply convergent for $\mu$ then, for any ${t\geq 1}$, $F^t$ is simply convergent and $F$ and $F^t$ have the same $\mu$-limit set. The Theorem~\ref{theo:fdeux} above gives an example of CA which is not simply convergent.

As shown by the following theorem, the simple convergence assumption simplifies the $\mu$-limit set as well as some decision problems on it (to be compared to Theorems~\ref{thm:sigma3} and~\ref{thm:hardnilpo}).

\begin{theorem}
  \label{thm:simplyconvergent}
  Let $\mu$ be a computable translation invariant measure.
  \begin{itemize}
  \item if $\A$ is simply convergent for $\mu$ then $L_\mu(\A)$ is a $\Sigma_2^0$ set;
  \item there exists a $\Pi_2^0$ predicate that characterizes $\mu$-nilpotent CA among simply convergent CA;
  \item the set of simply convergent CA is $\Pi_3^0$ and it is $\Pi_3$-hard when $\mu$ is the uniform Bernoulli measure.
  \end{itemize}
\end{theorem}
\begin{proof}
  If $\A$ is simply convergent for $\mu$, we have the following characterization of $L_\mu$:
  \[u\in L_\mu \Leftrightarrow \exists t_0, \exists \epsilon, \forall t>t_0 :\ F^t\mu([u]_0)>\epsilon.\]
  We deduce the first item of the theorem.

  $\A$ is not $\mu$-nilpotent exactly when there are two different words of equal size in $L_\mu$. With the hypothesis of simple convergence, it can be written:
  \[\exists u,v, |u|=|v|,u\not=v, \exists t_0, \exists \epsilon_u, \exists \epsilon_v, \forall t>t_0 :\ F^t\mu([u]_0)>\epsilon_u\wedge F^t\mu([v]_0)>\epsilon_v\]
  and the second item of the theorem follows directly.

  To show the third item, let us first remark that simple convergence can be expressed by a $\Pi_3^0$ formula saying that the sequence of probabilities of appearance along time of each word is a Cauchy sequence:
  \[\forall u\in A^\ast, \forall\epsilon>0, \exists N, \forall p,q > N:\ \bigl|\A^p\mu([u]_0) - \A^q\mu([u]_0)\bigr|\leq\epsilon.\]
  Finally, for $\Pi_3^0$-hardness it is sufficient to verify that a subset of the CA constructed in the proof of Proposition~\ref{prop:rice} are either $\mu$-nilpotent (hence simply convergent), or not simply convergent. More precisely, in the construction, consider a $\mu$-nilpotent CA $\A$ ($\Lmu(\A)=\std^*$ for some special state $\std$). First note that in Theorem~\ref{thm:maintheorem} the simple convergence of the CA is equivalent to the simple convergence of the densities of words in the computable sequence $(w_i)_i$. If you take a machine $\phi \notin COF$, then $\B$ is $\mu$-nilpotent as shown in Claim~\ref{cla:ricenoncof}. In the other case, with $\phi \in COF$, we still have ${\liminf_i d_{w_i}(\std) = 1}$. Indeed, you get the result with a sequence $(i_j)_j$ where $f(i_j)=(0,0)$. Hence, as in this case $L_w=Q_{\A}^*$ (Claim~\ref{cla:ricecof}), the convergence cannot be simple.

\end{proof}

Complexity considerations allow to prove that some $\mu$-limit sets are impossible to obtain with simple convergence. We currently do not know any direct proof of this fact.

\begin{corollary}
  There exists a CA whose $\mu$-limit set can not be the $\mu$-limit set of any simply convergent CA.
\end{corollary}
\begin{proof}
  By Theorem~\ref{thm:sigma3} there exists a $\Sigma_3^0$-hard $\mu$-limit set. However, Theorem~\ref{thm:simplyconvergent} shows that simply convergent CA produce $\mu$-limit sets which are only $\Sigma_2^0$.
\end{proof}

\subsection{Cesaro mean}

The construction from Theorem~\ref{thm:maintheorem} allowed us to
build complex $\mu$-limit sets. It also shows that this complexity
can be completely wiped out when considering the Cesaro mean.

\begin{theorem}
  For any cellular automaton, there exists another one with the same
  $\mu$-limit set (possibly up to a single uniform configuration) but which is $\mu$-Cesaro nilpotent.
\end{theorem}
\begin{proof}
  It's a direct application of Theorem~\ref{thm:maintheorem} where
  $w$ is chosen so that $L_w$ is the $\mu$-limit set of the given CA
  and $w'$ chosen so that $L_{w'}$ contains only $1$ letter.
\end{proof}

\subsection{Non-recursive convergence time}
 
Here we want to point out the fact that convergence to the $\mu$-limit language may actually be really late, in particular the next proposition states that the convergence rate may be slower than any recursive function.

\begin{prop}
Given an enumeration of Turing machines, denote $T_m$ the halting time of machine $m\in\N$ on input $0$. If $m$ does not halt on $0$, $T_m=0$.

There exists a cellular automaton $F$ (with $0\in Q_F$) such that:
\begin{itemize}
\item $\forall n\in\N, \exists t_n \geq \max\left\{T_m : m\leq n\right\}, F^{t_n}\mu([0])\geq \frac{1}{2n}$;
\item $0\notin \Lmu(F)$.
\end{itemize}

\end{prop}

\begin{proof}
  To prove it we use again Theorem~\ref{thm:maintheorem} with some growing computable sequence $w\in\mathbb{W}(Q)(0\in Q)$ that we will describe. Note first that, due to the construction used in the proof of Theorem~\ref{thm:maintheorem}, the first point of the theorem is implied by
  \[\forall n\in\N, \exists i_n \geq \max\left\{T_m : m\leq n\right\}, d_{w_{i_n}}(0)\geq \frac{1}{n}\]

  Take an  enumeration $f$ of the integers such that the preimage of any integer is infinite and $\forall i\in\N, f(i)\geq i$. For $i\in\N$  with $f(i)=m$, simulate the computation of machine $m$ with input $0$ during $i$ steps. As in the proofs of Theorem~\ref{thm:sigma3} and Proposition~\ref{prop:rice}, for each $m\in\N$, the smallest $i$ such that the computation reaches its end is said to be successful and failed in the other case. If $i\in f^{-1}(m)$ is successful, take  $w_i=(1^{m-1}0)^i$, else $w_i=1^{im}$.

 Thus $0$ has density $\frac{1}{m}$ in the writing layers of segments only when the simulation of the machine halts for the smallest $i\in f^{-1}(m)$, and $0$ otherwise. The two points of the result are now easily verified.

 For the first point, given $n\in\N$, consider $m_n$ such that $T_{m_n}=\max\left\{T_m :  m\leq n\right\}$. There exists $i_n\in f^{-1}(m_n)$ successful, hence $d_{w_{i_n}}(0)\geq \frac{1}{n}$. 
 
 For the second point, given $n\in\N$, there exists $i_n\in\N$ such that every machine $m\leq n$ halts in less than $i_n$ steps or never. Now, take $j_n\geq i_n$ such that every $m\leq n$ has been enumerated between $i_n$ and $j_n$. Thus, for every  $i\geq j_n$ the density of the word $0$ is  $d_{w_i}(0)\leq\frac{1}{n}$.

\end{proof}

\section{Recap of results}

In this section $\mu$ denotes the uniform Bernoulli measure. First we give comparative recap of complexity of properties or problems concerning limit sets and $\mu$-limit sets.

\newcommand\resref[2]{
  \begin{tabular}{c}
    #1\\{\small (see #2)}
  \end{tabular}
}

\begin{center}
\begin{tabular}{r|c|c}
\bfseries Problem or property  & \bfseries Limit Set & \bfseries $\mu$-Limit Set\\ \toprule
\textit{Being a singleton} & \resref{$\Sigma^0_1$-complete}{\cite{Kari-1992}} & \resref{$\Pi_3^0$-complete}{Thm.~\ref{thm:hardnilpo}}\\
 \midrule
\textit{Any non-trivial property} & \resref{$\Sigma_1^0$-hard}{\cite{Kari-1994}}  & \resref{$\Pi_3^0$-hard}{Thm.~\ref{thm:rice}} \\
 \midrule
\textit{Worst-case language} & \resref{$\Pi^0_1$-complete}{\cite{Hurd-1987}} & \resref{$\Sigma_3^0$-complete}{Thm.~\ref{thm:sigma3}}\\
 \midrule
\textit{Simplest configuration} & always uniform & \resref{can be $\alpha$-complex}{Cor.~\ref{coro:kolmo}}\\
 \midrule
\textit{Simplest quasi-periodicity} & always periodicity & \resref{can be not recursively bounded}{Cor.~\ref{coro:quasiper}}\\
 \bottomrule
\end{tabular}
\end{center}

Below is a recap on how the complexity of some problems is affected by adding hypotheses on the input CA.

\begin{center}
\begin{tabular}{r|c|c}
\bfseries Type of input CA  & \bfseries Worst $L_\mu$ & \bfseries $\mu$-Nilpotency\\ \toprule
\textit{General case} & \resref{$\Sigma_3^0$-complete}{Thm.~\ref{thm:sigma3}} & \resref{$\Pi_3^0$-complete}{Thm.~\ref{thm:hardnilpo}}\\
 \midrule
\textit{Equicontinuous} & \resref{$\Sigma_1^0$}{Thm.~\ref{thm:equica}}  & \resref{$\Sigma_2^0$}{Prop.~\ref{prop:equica}} \\
 \midrule
\textit{Simply convergent} & \resref{$\Sigma^0_2$}{Thm.~\ref{thm:simplyconvergent}} & \resref{$\Pi_2^0$}{Thm.~\ref{thm:simplyconvergent}} \\
 \bottomrule
\end{tabular}
\end{center}

  As shown in~\cite{Hellouin-Sablik-2013} it is certainly possible to generalize the results obtained here for large sets of measures (which was not the purpose of the present paper). In this context, it becomes relevant to consider the particular case of surjective CA. Indeed, as the uniform Bernoulli measure is preserved by surjective CA, the $\mu$-limit set is the full shift, but for another measure, the question is open.

 Naturally, the extension of these results can be discussed for higher dimensions. In particular, some of them should be reached given an equivalent construction in higher dimensions.

\section*{Acknowledgment}
We are grateful for the time spent by the anonymous referees on the
first version of this paper and for the incitative to write a better
version through their numerous comments.

\bibliographystyle{alpha}
\bibliography{paper_mulimit}

\end{document}